\documentclass[envcountsame]{llncs}
\pdfoutput=1

\usepackage[UKenglish]{babel}
\usepackage{microtype}
\usepackage{stmaryrd}
\usepackage{latexsym}
\usepackage{amssymb}
\usepackage{amsmath}
\usepackage{array}
\usepackage[all]{xy}
\usepackage{url}
\usepackage{xspace}
\usepackage{enumerate}
\usepackage{color}

\newcommand{\setsorts}{\mathcal{S}}
\newcommand{\Var}{\mathcal{V}}
\newcommand{\Rules}{\mathcal{R}}
\newcommand{\Pairs}{\mathcal{P}}

\newcommand{\TCapp}{\mathit{TCap}}
\newcommand{\QQ}{\mathcal{Q}}
\newcommand{\DP}{\mathsf{DP}}
\newcommand{\decpijl}{\Rightarrow}
\newcommand{\arrz}{\ensuremath{\rightarrow}}
\newcommand{\arr}[1]{\ensuremath{\rightarrow_{#1}}}
\newcommand{\arrin}[1]{\ensuremath{\rightarrow_{#1}^{\mathsf{in}}}}
\newcommand{\arrr}[1]{\arr{#1}^*}

\newcommand{\domain}{\mathsf{domain}}

\newcommand{\userules}{\mathit{U\!R}}
\newcommand{\formrules}{\mathit{F\!R}}
\newcommand{\baseformrules}{\formrules_{\mathsf{base}}}
\newcommand{\sformrules}{\mathit{S\!R}}
\newcommand{\tcapsformrules}{\mathit{S\!R_T}}

\newcommand{\implformrules}{\formrules_T}

\newcommand{\filteruserules}{\userules}
\newcommand{\tcapuserules}{\userules_T}

\newcommand{\Ce}{\mathcal{C}_\epsilon}
\newcommand{\atype}{\iota}
\newcommand{\btype}{\kappa}

\newcommand{\asort}{\iota}
\newcommand{\bsort}{\kappa}

\newcommand{\atypedec}{\sigma}

\newcommand{\afun}{f}
\newcommand{\bfun}{g}

\newcommand{\avar}{x}
\newcommand{\bvar}{y}
\newcommand{\cvar}{z}
\newcommand{\aterm}{s}
\newcommand{\bterm}{t}
\newcommand{\cterm}{q}
\newcommand{\dterm}{u}
\newcommand{\eterm}{v}

\newcommand{\asub}{\gamma}
\newcommand{\bsub}{\delta}
\newcommand{\csub}{\chi}
\newcommand{\symb}[1]{\mathtt{#1}}
\newcommand{\oo}{\mathtt{o}}
\newcommand{\dpsort}{\mathtt{dpsort}}

\newcommand{\up}[1]{#1^\sharp}
\newcommand{\FV}{\mathit{Var}}
\newcommand{\rijtje}[1]{[#1]}
\newcommand{\N}{\mathbb{N}}

\newcommand{\filter}{\overline{\pi}}
\newcommand{\pitriv}{\pi_{\mathcal{T}}}
\newcommand{\seq}[1]{\vec{#1}}
\newcommand{\flag}{\mathit{f}}
\newcommand{\minimal}{\mathsf{m}}
\newcommand{\all}{\mathsf{a}}
\newcommand{\formative}{\mathsf{formative}}
\renewcommand{\formative}{\mathsf{form}}
\newcommand{\arbitrary}{\mathsf{arbitrary}}
\newcommand{\innermost}{\mathsf{innermost}}

\newcommand{\collapsing}{\mathit{Cl}}
\newcommand{\noncollapsing}{\mathit{NC}}

\newcommand{\lijst}{\symb{LIST}}
\newcommand{\nat}{\symb{NAT}}
\newcommand{\nil}{\symb{Nil}}
\newcommand{\cons}{\symb{Cons}}
\newcommand{\nul}{\symb{O}}
\newcommand{\suc}{\symb{S}}

\newcommand{\Fa}{\symb{a}}

\newcommand{\Fb}{\symb{b}}
\newcommand{\Fc}{\symb{c}}
\newcommand{\Ff}{\symb{f}}

\newcommand{\Fg}{\symb{g}}
\newcommand{\Fh}{\symb{h}}

\newcommand{\suptermeq}{\unrhd}

\newcommand{\ack}{\symb{Ack}}
\newcommand{\random}{\symb{Rnd}}
\newcommand{\upd}{\symb{Upd}}
\newcommand{\mbig}{\symb{Big}}
\newcommand{\error}{\symb{Err}}
\newcommand{\return}{\symb{Return}}
\newcommand{\result}{\symb{RESULT}}
\newcommand{\run}{\symb{Run}}

\newcommand{\aprove}{\textsf{AProVE}\xspace}

\newcommand{\secshort}{Section}
\newcommand{\paragraaf}[1]{\paragraph{\textbf{#1}}}

\newcommand{\confreport}[2]{#1}    
\renewcommand{\confreport}[2]{#2} 

\title{First-Order Formative Rules%
\thanks{Support by EPSRC \& the Austrian Science
Fund (FWF) international
project I963.}}

\author{Carsten Fuhs\inst{1} \and Cynthia Kop\inst{2}}
\tocauthor{Carsten Fuhs (University College London), Cynthia Kop (University of Innsbruck)}
\institute{
  University College London,
  Dept.\ of Computer Science,
          London WC1E 6BT, UK
  \and
  University of Innsbruck, Institute of Computer Science,
          6020 Innsbruck, Austria
}

\begin{document}

\confreport{}{
\begin{center}
{\LARGE \textbf{First-Order Formative Rules}} \\[3ex]
{\large \textbf{Carsten Fuhs and Cynthia Kop}}
\end{center}
\definecolor{rectification}{rgb}{0.9, 0.9, 0.9}
\colorbox{rectification}{\parbox{12cm}{
\ \\
This technical report is an extended version of the paper
\cite{confversion}, to be published at RTA-TLCA 2014.
It includes an appendix for three purposes:
\begin{itemize}
\item giving complete proofs of the results which, for space reasons,
  are only sketched in the paper;
\item providing extensions to the results in the paper, which are
  used to obtain a more convenient implementation;
\item discussing implementation details.
\end{itemize}
\thispagestyle{empty}
}}
\newpage
\setcounter{page}{1}
}

\maketitle

\begin{abstract}
This paper discusses the method of \emph{formative rules} for
first-order term rewriting, which was previously defined for a
higher-order setting. Dual to the well-known \emph{usable rules},
formative rules allow dropping some of the term constraints that
need to be solved during a termination proof.
Compared to the higher-order definition, the first-order setting
allows for significant improvements of the technique.
\end{abstract}

\section{Introduction}\label{sec:intro}

In~\cite{kop:raa:11:1,kop:raa:12:1} C.\ Kop and F.\ van Raamsdonk
introduce the notion of \emph{formative rules}.  The technique is
similar to the method of
\emph{usable rules}~\cite{art:gie:00:1,gie:thi:sch:fal:06,hir:mid:07:1},
which is commonly used
in termination proofs, but has different strengths and weaknesses.

Since, by~\cite{tan:gal:91:1},
the more common \emph{first-order} style of term rewriting, both
with and without types, can be seen as a subclass of the formalism
of~\cite{kop:raa:12:1},
this result
immediately applies to first-order rewriting.  In an
untyped setting, we will, however, lose some of its strength, as
sorts play a relevant role
in 
formative rules.

On the other hand, by omitting the complicating aspects of
higher-order term rewriting (such as $\lambda$-abstraction and
``collapsing'' rules $l \arrz \avar \cdot \bvar$) we also gain
possibilities not present in the original setting; both things which
\emph{have not} been done, as the higher-order dependency pair
framework \cite{kop:12} is still rather limited, and things which
\emph{cannot} be done, at least with current theory.
Therefore, in this paper, \linebreak we will redefine the method for
(many-sorted) first-order term rewriting.

New compared to~\cite{kop:raa:12:1}, we will integrate formative
rules into the dependency pair framework~\cite{gie:thi:sch:05:2},
which is the basis of most contemporary termination provers for
first-order term rewriting.
Within this framework, formative rules are used either as a
stand-alone processor or with reduction pairs, and can be coupled
with usable rules and argument filterings.
We also formulate a semantic characterisation of formative rules,
to enable future generalisations of the definition.
Aside from this, we present a (new) way to weaken the detrimental
effect of collapsing rules.

This paper is organised as follows.
After 
the
preliminaries in \secshort~\ref{sec:prelim}, a first
definition of formative rules is given
and then generalised
in \secshort~\ref{sec:formrules}.
\secshort~\ref{sec:procs} shows various ways to use formative rules in 
the dependency pair frame\-work.
\secshort~\ref{sec:unsorted} 
gives
an alternative way to deal with
collapsing rules.
In \secshort~\ref{sec:innermost} we\confreport{}{\linebreak} consider innermost termination,
\secshort~\ref{sec:code} describes implementation and
experiments, and in
\secshort~\ref{sec:future} we
 point out possible
future work and conclude.
All proofs and an improved formative rules approximation are provided
in~\confreport{\cite{techreportversion}}{the appendix}.

\section{Preliminaries}\label{sec:prelim}

We consider \emph{many-sorted term rewriting}: term
rewriting with \emph{sorts}, basic types.
While sorts are not
usually considered in studies of first-order
term rewrite systems (TRSs) and
for instance the Termination Problems Data Base\footnote{More information
  on the \emph{TPDB}: \url{http://termination-portal.org/wiki/TPDB}}
does
not include them (for first-order TRSs),\footnote{This
may also be
  due
  to the fact that currently most termination tools for first-order
  rewriting only make very limited use of the additional information
  carried by types.} they are a natural
addition; in typical applications there is little
need to allow untypable terms like $\symb{3} + \symb{apple}$.
Even when no sorts are\linebreak present, a standard
TRS can be seen as a many-sorted TRS with only one
sort.\footnote{However, the method of this paper is stronger given
more sorts.  We may be able to (temporarily) infer richer sorts,
however.  We will say more about this in \secshort~\ref{sec:innermost}.}

\paragraaf{Many-sorted TRSs}
We assume given a non-empty set $\setsorts$ of \emph{sorts}; these
are typically things like $\mathtt{Nat}$ or $\mathtt{Bool}$, or (for
representing unsorted systems) $\setsorts$ might be the set with a
single sort $\{ \oo \}$.
A \emph{sort declaration} is a sequence $[\btype_1 \times \ldots
\times \btype_n] \decpijl \atype$ where $\atype$ and all $\btype_i$
are sorts.  A sort declaration $[] \decpijl \atype$ is just
denoted $\atype$.

A \emph{many-sorted signature} is a set $\Sigma$ of function symbols
$\afun$, each equipped with a sort declaration $\atypedec$, notation
$\afun : \atypedec \in \Sigma$.  Fixing a many-sorted signature
$\Sigma$ and an infinite set $\Var$ of sorted variables,
the set of \emph{terms}
consists of those expressions $\aterm$ over $\Sigma$ and $\Var$ for
which we can derive $\aterm : \atype$ for some sort $\atype$, using
the clauses:
\[
\begin{array}{rcl}
\avar : \atype & \mathrm{if} & \avar : \atype \in \Var \\
\afun(\aterm_1,\ldots,\aterm_n) : \atype & \mathrm{if} &
  \afun : [\btype_1 \times \ldots \times \btype_n] \decpijl \atype \in
  \Sigma\ \mathrm{and}\ \aterm_1 : \btype_1,\ \ldots,\ \aterm_n :
  \btype_n \\
\end{array}
\]
We often denote $\afun(\aterm_1,\ldots,\aterm_n)$ as just
$\afun(\seq{\aterm})$.
Clearly, every term has a unique sort.\linebreak
Let $\FV(\aterm)$ be the set of all variables occurring in a 
term $\aterm$.
A term $\aterm$ is \emph{linear} if every variable in
$\FV(\aterm)$ occurs only once in $\aterm$.
A term $\bterm$ is a \emph{subterm} of another term $\aterm$, notation
$\aterm \suptermeq \bterm$, if either $\aterm = \bterm$ or $\aterm =
\afun(\aterm_1,\ldots,\aterm_n)$ and some $\aterm_i \suptermeq
\bterm$.
A \emph{substi-\linebreak tution} $\gamma$ is a mapping from variables to terms of
the same sort; the application $\aterm\gamma$ of a substitution
$\gamma$ on a term $\aterm$ is $\aterm$ with each $\avar
\in \domain(\gamma)$ replaced by $\gamma(\avar)$.

A \emph{rule} is a pair $\ell \arrz r$ of terms with the same sort
such that $\ell$ is not a variable.%
\footnote{Often also $\FV(r) \subseteq \FV(\ell)$ is required.
However, we use
\emph{filtered} rules $\filter(\ell) \arrz \filter(r)$ later, where
the restriction is inconvenient.
As a rule is non-terminating if $\FV(r) \not \subseteq \FV(\ell)$,
as usual we
forbid such rules in the input
$\Rules$
and in
dependency pair problems.
}
A rule is \emph{left-linear} if $\ell$ is linear, and
\emph{collapsing} if $r$ is a va\-riable.
Given a set of rules $\Rules$, the \emph{reduction relation}
$\arr{\Rules}$ is
given by:
$\ell\gamma \arr{\Rules} r\gamma$ if $\ell \arrz r \in \Rules$ and
  $\gamma$ a substitution;
$\afun(\ldots,\aterm_i,\ldots) \arr{\Rules}
  \afun(\ldots,\aterm_i',\ldots)$ if
  $\aterm_i \arr{\Rules} \aterm_i'$.
A term $\aterm$ is in \emph{normal form} if there is no $\bterm$ such
that $\aterm \arr{\Rules} \bterm$.

The relation $\arrr{\Rules}$ is the transitive-reflexive closure of
$\arr{\Rules}$.  If there is a rule $\afun(\seq{l})
\arrz r \in \Rules$ we say that
$\afun$ is a \emph{defined symbol}; otherwise $\afun$ is a
\emph{constructor}.

A \emph{many-sorted term rewrite system (MTRS)} is
a pair
$(\Sigma,\Rules)$
with signature $\Sigma$\linebreak
and  a set $\Rules$
of rules $\ell \arrz r$ with $\FV(r) \subseteq \FV(\ell)$.
A term $\aterm$ is \emph{terminating} if there is\linebreak no infinite
reduction $\aterm \arr{\Rules} \bterm_1 \arr{\Rules} \bterm_2
\ldots$
An MTRS is terminating if all terms are.

\begin{example}\label{ex:running}
An example of a many-sorted TRS $(\Sigma,\Rules)$ with more than one sort is the
following system, which uses lists, natural numbers and a
$\result$ sort:\pagebreak
\[
\begin{array}{rclrclrcl}
\nul & : & \nat &
\cons & : & [\nat \times \lijst] \decpijl \lijst &
\run & : & [\lijst] \decpijl \result \\
\suc & : & [\nat] \decpijl \nat &
\ack & : & [\nat \times \nat] \decpijl \nat &
\return & : & [\nat] \decpijl \result \\
\nil & : & \lijst &
\mbig & : & [\nat \times \lijst] \decpijl \nat &
\random & : & [\nat] \decpijl \nat \\
\error & : & \result & 
\upd & : & [\lijst] \decpijl \lijst \\
\end{array}
\]
\[
\begin{array}{rrclrrcl}
1. & \random(\avar) & \arrz & \avar &
6. & \mbig(\avar,\nil) & \arrz & \avar \\
2. & \random(\suc(\avar)) & \arrz & \random(\avar) &
7. & \mbig(\avar,\cons(\bvar,\cvar)) & \arrz & \mbig(\ack(\avar,\bvar),\upd(\cvar)) \\
3. & \upd(\nil) & \arrz & \nil &
8. & \upd(\cons(\avar,\bvar)) & \arrz & \cons(\random(\avar),\upd(\bvar)) \\
4. & \run(\nil) & \arrz & \error &
9. & \run(\cons(\avar,\bvar)) & \arrz & \return(\mbig(\avar,\bvar)) \\
5. & \ack(\nul,\bvar) & \arrz & \suc(\bvar) &
10. & \ack(\suc(\avar),\bvar) & \arrz & \ack(\avar,\suc(\bvar)) \\
& & & & 
11. & \ack(\suc(\avar),\suc(\bvar)) & \arrz & \ack(\avar,\ack(\suc(\avar),\bvar)) \\
\end{array}
\]
$\run(\mathit{lst})$ calculates a potentially very large number,
depending on the elements of $\mathit{lst}$ and some randomness.  We
have chosen this example because it will help to demonstrate
the various aspects of formative rules, without being too long.
\end{example}

\paragraaf{The Dependency Pair Framework}
As a basis to study termination, we will use the \emph{dependency
pair (DP) framework}~\cite{gie:thi:sch:05:2}, adapted to include sorts.

Given an MTRS $(\Sigma,\Rules)$, let $\Sigma^\sharp = \Sigma
\cup \{ \up{\afun} : [\atype_1 \times \ldots \times \atype_n]
\decpijl \dpsort \mid \afun : [\atype_1 \times \ldots \times \atype_n]
\decpijl \btype \in \Sigma \;\!\land\;\! \afun$ a defined symbol of $\Rules\}$,
where $\dpsort$ is a fresh sort.
The set\linebreak $\DP(\Rules)$ of \emph{dependency pairs (DPs)} of $\Rules$
consists of all rules of the form $\up{\afun}(l_1,\ldots,l_n)\linebreak \arrz
\up{\bfun}(r_1,\ldots,r_m)$ where $\afun(\seq{l}) \arrz r \in \Rules$
and $r \suptermeq \bfun(\seq{r})$ with $\bfun$ a defined symbol.

\begin{example}\label{ex:DPs}
The dependency pairs of the system in Example~\ref{ex:running} are:
\[
\begin{array}{rclrcl}
\up{\random}(\suc(\avar)) & \arrz & \up{\random}(\avar) &
\up{\mbig}(\avar,\cons(\bvar,\cvar)) & \arrz & \up{\mbig}(\ack(\avar,\bvar),\upd(\cvar)) \\
\up{\upd}(\cons(\avar,\bvar)) & \arrz & \up{\random}(\avar) &
\up{\mbig}(\avar,\cons(\bvar,\cvar)) & \arrz & \up{\ack}(\avar,\bvar) \\
\up{\upd}(\cons(\avar,\bvar)) & \arrz & \up{\upd}(\bvar) &
\up{\mbig}(\avar,\cons(\bvar,\cvar)) & \arrz & \up{\upd}(\cvar) \\
\up{\run}(\cons(\avar,\bvar)) & \arrz & \up{\mbig}(\avar,\bvar) &
\up{\ack}(\suc(\avar),\suc(\bvar)) & \arrz & \up{\ack}(\avar,\ack(\suc(\avar),\bvar)) \\
\up{\ack}(\suc(\avar),\bvar) & \arrz & \up{\ack}(\avar,\suc(\bvar)) &
\up{\ack}(\suc(\avar),\suc(\bvar)) & \arrz & \up{\ack}(\suc(\avar),\bvar) \\
\end{array}
\]
\end{example}

For sets $\Pairs$ and $\Rules$ of rules, an
infinite $(\Pairs,\Rules)$-chain is a sequence
$[(\ell_i \arrz r_i,\linebreak\gamma_i) \mid i \in
\N]$ where each $\ell_i \arrz r_i \in \Pairs$ and $\gamma_i$ is a
substitution such that $r_i\gamma_i \arrr{\Rules} \ell_{i+1}
\gamma_{i+1}$.  This chain is \emph{minimal} if each $r_i\gamma_i$ is
terminating with respect to $\arr{\Rules}$.

\begin{theorem} (following \cite{art:gie:00:1,gie:thi:sch:05:2,gie:thi:sch:fal:06,hir:mid:07:1})\label{thm:dps}
An MTRS $(\Sigma,\Rules)$ is terminating if and only if
there is no infinite minimal $(\DP(\Rules),\Rules)$-chain.
\end{theorem}

A 
\emph{DP problem}
is
a 
triple
$(\Pairs,\Rules,\flag)$ with $\Pairs$ and $\Rules$ sets of
rules and $\flag \in \{\minimal,\all\}$ 
(denoting
$\{$\textsf{m}inimal, \textsf{a}rbitrary$\}$).\footnote{Here we do
not modify the signature $\Sigma^\sharp$ of a DP problem, so
we leave $\Sigma^\sharp$ implicit.}
A DP problem $(\Pairs,\Rules,\flag)$ is \emph{finite} if there is no
infinite $(\Pairs,\Rules)$-chain, which is minimal if $\flag =
\minimal$.
A \emph{DP processor} is a function which 
maps a DP
problem 
to a set of DP problems.
A processor $\mathit{proc}$ is \emph{sound} if, for all
DP problems $A$: if 
all
$B \in\mathit{proc}(A)$ are finite, then $A$ is finite.

The goal of the DP framework is,
starting with a set $D = \{(\DP(\Rules),\Rules,\minimal)\}$,
to reduce $D$ 
to $\emptyset$ using sound processors.
Then we may
conclude termination of the initial MTRS $(\Sigma,\Rules)$.%
\footnote{The full DP framework \cite{gie:thi:sch:05:2} can also be used
for proofs of \emph{non-termination}. Indeed, by
\cite[Lemma 2]{gie:thi:sch:05:2},
all processors
introduced in this paper
(except
Theorem~\ref{thm:processor}
for innermost rewriting)
are ``complete'' and
may be applied 
in a
non-termination proof.
}
Various common processors \pagebreak 
use a
\emph{reduction pair}, 
a pair $(\succsim,\succ)$ of a monotonic, stable (closed under substitutions)
quasi-ordering $\succsim$ on terms
and a well-founded, stable ordering $\succ$ compatible with
$\succsim$ (i.e., $\succ \cdot \succsim\ \subseteq\ 
\succ$).

\begin{theorem} (following \cite{art:gie:00:1,gie:thi:sch:05:2,gie:thi:sch:fal:06,hir:mid:07:1})\label{thm:redpair}
Let $(\succsim,\succ)$ be a reduction pair.
The processor which
maps a DP problem
$(\Pairs,\Rules,\flag)$
to the following result is sound:
\begin{itemize}
\item $\{(\Pairs \setminus \Pairs^\succ,\Rules,\flag)\}$ if:
  \begin{itemize}
  \item $\ell \succ r$ for $\ell \arrz r \in \Pairs^\succ$ and
  $\ell \succsim r$ for $\ell \arrz r \in \Pairs \setminus \Pairs^\succ$
    (with $\Pairs^\succ \subseteq \Pairs$);
  \item $\ell \succsim r$ for $\ell \arrz r \in \Rules$.
  \end{itemize}
\item $\{(\Pairs, \Rules, \flag)\}$ otherwise
\end{itemize}
\end{theorem}

Here,
we must orient all elements of $\Rules$ with
$\succsim$.  As there are many processors which remove elements from
$\Pairs$ and few which remove 
from $\Rules$, this may give many\linebreak
constraints.  \emph{Usable rules}, often combined with
\emph{argument filterings}, address this:

\begin{definition}\label{def:ur} (following \cite{gie:thi:sch:fal:06,hir:mid:07:1})
Let $\Sigma$ be a signature and $\Rules$ a set of rules.
An \emph{argument filtering} is a function that maps each
$\afun : [\asort_1 \times \ldots \times \asort_n] \decpijl \bsort$ to
a set $\{i_1,\ldots,i_k\} \subseteq \{1,\ldots,n\}$.%
\footnote{Usual definitions of argument filterings also allow
$\pi(\afun) = i$, giving $\filter(\afun(\seq{\aterm})) =
\filter(\aterm_i)$, but for usable rules, $\pi(\afun) =
i$ is treated 
the same
as $\pi(\afun) = \{ i \}$, cf.\ 
\cite[\secshort~4]{gie:thi:sch:fal:06}.
}
The \emph{usable rules} of a term $\bterm$ with respect to an
argument filtering $\pi$ are defined as the smallest set
$\userules(\bterm,\Rules,\pi) \subseteq \Rules$ such that:
\begin{itemize}
\item if $\Rules$ is not finitely branching
  (i.e.\ there are terms with infinitely many direct reducts),
  then $\userules(\bterm,\Rules,\pi) = \Rules$;
\item if $\bterm = \afun(\bterm_1,\ldots,\bterm_n)$, then
  $\userules(\bterm_i,\Rules,\pi) \subseteq
  \userules(\bterm,\Rules,\pi)$ for all $i \in \pi(\afun)$;
\item if $\bterm = \afun(\bterm_1,\ldots,\bterm_n)$, then
  $\{ \ell \to r \in \Rules \mid 
  \ell = \afun(\ldots) \}
  \subseteq \userules(\bterm,\Rules,\pi)$;
\item if $\ell \arrz r \in \userules(\bterm,\Rules,\pi)$, then
  $\userules(r,\Rules,\pi) \subseteq \userules(\bterm,\Rules,\pi)$.
\end{itemize}
For a set of rules $\Pairs$, we define $\userules(\Pairs,\Rules,\pi) =
\bigcup_{s \to t \in \Pairs} \userules(t,\Rules,\pi)$.
\end{definition}

Argument filterings $\pi$ are used to disregard arguments of certain
function\linebreak symbols.  Given $\pi$, let $\afun_\pi : [\atype_{i_1} \times
\ldots \times \atype_{i_k}] \decpijl \btype$ be a fresh function
symbol for all $\afun$\linebreak  with $\pi(\afun) = \{i_1,\ldots,i_k\}$
and $i_1 < \ldots < i_k$,
and define
  $\filter(\avar) = \avar$ for $\avar$ a variable,\linebreak  and
  $\filter(\afun(\aterm_1,\ldots,\aterm_n)) =
  \afun_\pi(\filter(\aterm_{i_1}),\ldots,\filter(\aterm_{i_k}))$
  if $\pi(\afun) = \{i_1,\ldots,i_k\}$ and $i_1 < \ldots < \linebreak i_k$.
For a set of rules $\Rules$, let $\filter(\Rules) = \{
\filter(l) \arrz \filter(r) \mid l \arrz r \in \Rules \}$.
The idea of usable rules is 
to only
consider rules
relevant to the 
pairs in $\Pairs$ after applying
$\filter$.

Combining usable rules, argument filterings and reduction pairs, we
obtain:

\begin{theorem}\label{thm:userules} (\cite{gie:thi:sch:fal:06,hir:mid:07:1})
Let $(\succsim,\succ)$ be a reduction pair and
$\pi$ an argument filtering.
The processor which
maps
a DP problem
$(\Pairs,\Rules,\flag)$
to
the following result is 
sound:
\begin{itemize}
\item $\{(\Pairs \setminus \Pairs^\succ,\Rules,\minimal)\}$ if $\flag = \minimal$ and:
  \begin{itemize}
  \item $\filter(\ell) \succ \filter(r)$ for $\ell \arrz r \in \Pairs^\succ$ and
  $\filter(\ell) \succsim \filter(r)$ for $\ell \arrz r \in \Pairs \setminus \Pairs^\succ$;
  \item $\filter(\ell) \succsim \filter(r)$ for $\ell \arrz r \in \userules(\Pairs,\Rules,\pi) \cup \Ce$, \\
    where $\Ce = \{ \Fc_\asort(\avar,\bvar) \arrz \avar, \Fc_\asort(\avar,\bvar) \arrz \bvar \mid$ all sorts $\asort \}$.
  \end{itemize}
\item $\{(\Pairs, \Rules, \flag)\}$ otherwise
\end{itemize}
\end{theorem}

We define $\userules(\Pairs,\Rules)$ as $\userules(\Pairs,\Rules,
\pitriv)$, where $\pitriv$ is the \emph{trivial 
filtering}:
$\pitriv(\afun)
= \{1,\ldots,n\}$ for $\afun : [\asort_1 \times \ldots
\times \asort_n] \decpijl \bsort \in \Sigma$.
Then Theorem~\ref{thm:userules} is exactly the
standard
reduction pair processor, but with constraints on
$\userules(\Pairs,\Rules) \cup \Ce$ instead of $\Rules$.  We could
also 
use
a processor which maps $(\Pairs,\Rules,\minimal)$ to
$\{(\Pairs,\userules(\Pairs,\Rules) \cup \Ce,\all)\}$, but as this 
loses the minimality flag, it is usually not a good idea (various
processors need this flag,
including usable rules!) and can only be done once.\pagebreak

\section{Formative Rules}\label{sec:formrules}

Where
\emph{usable rules}~\cite{art:gie:00:1,gie:thi:sch:fal:06,hir:mid:07:1} 
are defined 
primarily by the right-hand sides of $\Pairs$ and $\Rules$, the
\emph{formative rules} discussed here are defined by the left-hand
sides.
This has\linebreak
consequences; most importantly, we cannot handle
non-left-linear rules very well.

We fix a signature $\Sigma$.  
A
term $\aterm : \atype$ \emph{has shape} $\afun$
with $\afun :
  [\seq{\btype}]
  \decpijl \atype \in \Sigma$
if either
  $\aterm = \afun(r_1,\ldots,r_n)$, or
  $\aterm$ is a variable of sort $\atype$.
That is, there exists 
some
$\asub$ with $\aterm\asub =
\afun(\ldots)$: one can \emph{specialise} $\aterm$ to have $\afun$ as
its root symbol.

\begin{definition}\label{def:form}
Let $\Rules$ be a set of rules.
The \emph{basic formative rules} of a term $\bterm$ are defined as the
smallest set $\baseformrules(\bterm,\Rules) \subseteq \Rules$ such that:
\begin{itemize}
\item if $\bterm$ is not linear, then $\baseformrules(\bterm,\Rules) =
  \Rules$;
\item if $\bterm = \afun(\bterm_1,\ldots,\bterm_n)$, then
  $\baseformrules(\bterm_i,\Rules) \subseteq \baseformrules(\bterm,
  \Rules)$;
\item if $\bterm = \afun(\bterm_1,\ldots,\bterm_n)$, then
  $\{ \ell \to r \in \Rules \mid r$ has shape $\afun \} \subseteq
  \baseformrules(\bterm,\Rules)$;
\item if $\ell \arrz r \in \baseformrules(\bterm,\Rules)$, then
  $\baseformrules(\ell,\Rules) \subseteq \baseformrules(\bterm,\Rules)$.
\end{itemize}
For rules $\Pairs$, let $\baseformrules(\Pairs,\Rules) =
\bigcup_{s \to t \in \Pairs} \baseformrules(s,\Rules)$.
Note that $\baseformrules(\avar,\!\Rules) = \emptyset$.
\end{definition}

Note the strong symmetry with Definition~\ref{def:ur}.  We have
omitted the argument filtering $\pi$ here, because the definitions
are simpler without it.  In Section~\ref{sec:procs} we will see how
we can add argument filterings back in without changing the definition.

\begin{example}\label{ex:formrules}\label{ex:runningP}
In the system from Example~\ref{ex:running}, consider
$\Pairs = \{ \up{\mbig}(\avar,\cons(\bvar,\cvar)) \arrz
\up{\mbig}(\ack(\avar,\bvar),\upd(\cvar))\}$.  The symbols
in the left-hand side are just $\up{\mbig}$ (which has sort
$\dpsort$, which is not used in $\Rules$) and $\cons$.  Thus,
$\baseformrules(\Pairs,\Rules) = \{ 8 \}$.
\end{example}

Intuitively, the formative rules of a dependency pair $\ell \arrz r$
are those rules which might contribute to creating the pattern $\ell$.
In Example~\ref{ex:formrules}, to reduce a term
$\up{\mbig}(\ack(\suc(\nul),\nul),\upd(\cons(\nul,\nil)))$ to an
instance of $\up{\mbig}(\avar,
\cons(\bvar,\cvar))$,
a
single step with the $\upd$ rule 8
gives
$\up{\mbig}(\ack(\suc(\nul),\nul),\cons(\random(\nul),\upd(\nil)))$;
we 
need not
reduce the
$\ack()$ or $\random()$ subterms 
for this.
To create a \emph{non}-linear pattern, any rule could contribute, as a step
deep inside a term may be needed.

\begin{example}
Consider $\Sigma = \{ \Fa,\Fb : \symb{A}, \up{\Ff} : [\symb{B}
\times \symb{B}] \decpijl \dpsort, \Fh : [\symb{A}] \decpijl
\symb{B} \}$,\ $\Rules = \{ \Fa \to \Fb \}$ and
$\Pairs = \{ \up{\Ff}(\avar,\avar) \arrz \up{\Ff}(\Fh(\Fa),\Fh(\Fb))\}$.
Without the linearity restriction,
$\baseformrules(\Pairs,\Rules)$ would be $\emptyset$, as $\dpsort$ does
not occur in the rules and $\baseformrules(\avar,\Rules) = \emptyset$.
But there is no infinite $(\Pairs,\emptyset)$-chain, while we
do have an infinite $(\Pairs,\Rules)$-chain, with $\gamma_i =
[\avar:=\Fh(\Fb)]$ for all $i$. 
The $\Fa \to \Fb$ rule
is needed to
make $\Fh(\Fa)$ and $\Fh(\Fb)$ equal.
Note that this happens even though
the sort of
$\avar$ does not occur in $\Rules$!
\end{example}

Thus, as we will see, in an infinite $(\Pairs,\Rules)$-chain we can
limit interest to rules in $\baseformrules(\Pairs,\Rules)$.  We call
these \emph{basic} formative rules because while they demonstrate the
concept, in practice we would typically use more advanced
extensions of the idea.  For instance, following the $\TCapp$ idea
of~\cite[Definition 11]{gie:thi:sch:05:1}, a rule $l \arrz
\afun(\nul)$ does not need to be a formative rule of
$\afun(\suc(\avar)) \arrz r$ if $\nul$ is a constructor.

To use formative rules with 
DPs,
we will show that any
$(\Pairs,\Rules)$-chain can be altered so that the $r_i\gamma_i
\arrr{\Rules} \ell_{i+1}\gamma_{i+1}$ reduction has a very specific
form (which uses only formative rules of $\ell_{i+1}$).
To this end, we consider \emph{formative reductions}.  A formative
reduction is a reduction where, essentially, a rewriting step is
only done if it is needed to obtain a result of the right form.

\begin{definition}[Formative Reduction]
For a term $\ell$, substitution $\gamma$
and
term $\aterm$, we say $\aterm \arrr{\Rules} \ell\gamma$ by a
\emph{formative $\ell$-reduction} if one of the following holds:
\begin{enumerate}
\item $\ell$ is non-linear;
\item\label{it:variablecase}
  $\ell$ is a variable and $\aterm = \ell\gamma$;
\item $\ell = \afun(l_1,\ldots,l_n)$ and $\aterm = \afun(\aterm_1,
  \ldots,\aterm_n)$ and each $\aterm_i \arrr{\Rules} l_i\gamma$ by a
  formative $l_i$-reduction;
\item $\ell = \afun(l_1,\ldots,l_n)$ and there are a rule $\ell'
  \arrz r' \in \Rules$ and a substitution $\delta$ such that $\aterm
  \arrr{\Rules} \ell'\delta$ by a formative $\ell'$-reduction and
  $r'\delta = \afun(\bterm_1,\ldots,\bterm_n)$ and each $\bterm_i
  \arrr{\Rules} l_i\gamma$ by a formative $l_i$-reduction.
\end{enumerate}
\end{definition}

\noindent
Point~\ref{it:variablecase} is
the
key: a reduction $\aterm \arrr{\Rules}
\avar\gamma$ must be postponed.
Formative reductions
are the base of a semantic definition of formative rules:

\begin{definition}\label{def:semantic}
A function $\formrules$ that maps a term $\ell$ and a set of rules
$\Rules$ to a set $\formrules(\ell,\Rules) \subseteq \Rules$ is a
\emph{formative rules approximation} if for all $\aterm$ and $\gamma$: 
if $\aterm \arrr{\Rules} \ell\gamma$ by a formative $\ell$-reduction,
then this reduction uses only rules in $\formrules(\ell,\Rules)$.

Given a formative rules approximation $\formrules$, let
$\formrules(\Pairs,\Rules) = 
\bigcup_{s \to t \in \Pairs} \formrules(s,\Rules)$.
\end{definition}

As might be expected, $\baseformrules$ is indeed a formative rules
approximation:

\begin{lemma}\label{lem:redtorules}
A formative $\ell$-reduction $\aterm \arrr{\Rules} \ell\gamma$ uses
only rules in $\baseformrules(\ell,\Rules)$.
\end{lemma}

\begin{proof}
By induction on the definition of a formative $\ell$-reduction.
If $\ell$ is non-linear, then $\baseformrules(\ell,\Rules) = \Rules$,
so this is clear.
If $\aterm = \ell\gamma$ then no rules play a part.

If $\aterm = \afun(\aterm_1,\ldots,\aterm_n)$ and $\ell = \afun(l_1,
\ldots,l_n)$ and each $\aterm_i \arrr{\Rules} l_i\gamma$ by a
formative $l_i$-reduction, then by the induction hypothesis each
formative $l_i$-reduction 
$\aterm_i \arrr{\Rules} l_i\gamma$ uses only rules in
$\baseformrules(l_i,\Rules)$.  Observing that by definition
$\baseformrules(l_i,\Rules) \subseteq \baseformrules(\ell,\Rules)$,
we see that all steps of the reduction use rules in
$\baseformrules(\ell,\Rules)$.

If $\aterm \arrr{\Rules} \ell'\delta \arr{\Rules} r'\delta =
\afun(\bterm_1,\ldots,\bterm_n) \arrr{\Rules} \afun(l_1,\ldots,l_n)
\gamma = \ell\gamma$, then by the same reasoning the reduction
$r'\delta \arrr{\Rules} \ell\gamma$ uses only formative rules of
$\ell$, and by the induction hypothesis $\aterm \arrr{\Rules} \ell'
\delta$ uses only formative rules of $\ell'$.
Noting that $r'$ obviously has the same sort as $\ell$, and either
$r'$ is a variable or a term 
$\afun(r'_1,\ldots,r'_n)$,
we see that $r'$ has shape $\afun$, so $\ell' \arrz r' \in
\baseformrules(\ell,\Rules)$.  Therefore $\baseformrules(\ell',\Rules)
\subseteq \linebreak
\baseformrules(\ell,\Rules)$, so all rules in the reduction are
formative rules of $\ell$.
\qed
\end{proof}

In the following, we will assume a fixed formative rules approximation
$\formrules$.  The relevance of formative rules is clear from their
definition: if we can prove that a $(\Pairs,\Rules)$-chain can be
altered to use formative reductions in the $\arr{\Rules}$ steps, then
we can drop all non-formative rules from a DP problem.

The key result in this paper is the following technical lemma, which
allows us to alter a reduction $\aterm \arrr{\Rules} \ell\gamma$ to a
formative reduction (by changing $\gamma$):

\newcommand{\arrpar}[1]{\longrightarrow\!\!\!\!\!\!\!\!\|~_{#1}\ }
\newcommand{\arrrpar}[2]{\arrpar{#1}\!\!\!^{#2}~}

\begin{lemma}\label{lem:base}
If $\aterm \arrr{\Rules} \ell\asub$ for some
terms $\aterm,\ell$ and a substitution $\asub$ on domain $\FV(\ell)$,
then 
there is a substitution $\bsub$ on the same domain such that
$\aterm \arrr{\formrules(\ell,\Rules)} \ell\bsub$ by a formative
$\ell$-reduction.
\end{lemma}

\begin{proof}
For non-linear $\ell$ this is clear, choosing $\delta := \gamma$.
So let $\ell$ be a linear term.
By definition of $\formrules$, it suffices to see that $\aterm
\arrr{\Rules} \ell\bsub$ by a formative $\ell$-reduc\-tion.  This
follows from the following claim:
\emph{If $\aterm \arrrpar{\Rules}{k} \ell\asub$ for some $k$, term
$\aterm$, linear term $\ell$ and substitution $\asub$ on domain
$\FV(\ell)$, then there is a substitution $\bsub$ on
$\FV(\ell)$
such that $\aterm \arrr{\Rules} \ell\bsub$
by a formative $\ell$-reduction, and each $\bsub(\avar)
\arrrpar{\Rules}{k} \asub(\avar)$.}\pagebreak

Here, the parallel reduction relation $\arrpar{\Rules}$ is defined
by:
  $\avar \arrpar{\Rules} \avar$;
  $\ell\gamma \arrpar{\Rules} r\gamma$ for $\ell \arrz r \in \Rules$;
  if $\aterm_i \arrpar{\Rules} \bterm_i$ for $1 \leq i \leq n$, then
  $\afun(\aterm_1,\ldots,\aterm_n) \arrpar{\Rules} \afun(\bterm_1,
  \ldots,\bterm_n)$.
The notation $\arrrpar{\Rules}{k}$ indicates $k$ \emph{or fewer} successive
$\arrpar{\Rules}$ steps.
Note that $\arrpar{\Rules}$ is reflexive, and if each $\aterm_i
\arrrpar{\Rules}{N_i} \bterm_i$, then
$\afun(\seq{\aterm}) \arrrpar{\Rules}{\max(N_1,\ldots,N_n)}
\afun(\seq{\bterm})$.

We prove the claim by induction first on $k$, second on the size of
$\ell$.

If $\ell$ is a variable we are immediately done, choosing $\delta :=
[\ell:=\aterm]$.

Otherwise, let $\ell = \afun(l_1,\ldots,l_n)$ and
$\gamma = \gamma_1 \cup \ldots \cup \gamma_n$ such that all $\gamma_i$
have disjoint domains and each $l_i\gamma_i = l_i\gamma$; this is
possible due to linearity.

First suppose the reduction $\aterm \arrrpar{\Rules}{k} \ell\asub$
uses no topmost steps.  Thus, we can write $\aterm = \afun(\aterm_1,
\ldots,\aterm_n)$ and each $\aterm_i \arrrpar{\Rules}{k} l_i\asub$.
By the second induction hypothesis we can find
$\bsub_1,\ldots,\bsub_n$ such that each $\aterm_i \arrr{\Rules} l_i
\bsub_i$ by a formative $l_i$-reduction and each $\bsub_i(\avar)
\arrrpar{\Rules}{k} \asub_i(\avar)$.
Choose $\bsub := \bsub_1 \cup \ldots \cup \bsub_n$; this is well-defined
by the assumption on the disjoint domains.
Then $\aterm \arrr{\Rules} \ell\bsub$ by a formative $\ell$-reduction.

Alternatively, a topmost step was done, which cannot be parallel with
other steps: $\aterm \arrrpar{\Rules}{m} \ell'\asub'
\arr{\Rules} r'\asub' \arrrpar{\Rules}{k-m-1} \ell\asub$ for some
$\ell' \arrz r' \in \Rules$ and substitution $\asub'$; we can safely
assume that $r'\asub' \arrrpar{\Rules}{k-m-1} \ell\asub$ does not use
topmost steps (otherwise we could just choose a later step).
Since $m < k$, the first induction hypothesis provides $\bsub'$ such
that $\aterm \arrr{\Rules} \ell'\bsub'$ by a formative
$\ell'$-reduction and each $\bsub'(\avar) \arrrpar{\Rules}{m}
\asub'(\avar)$.  But then also $r'\bsub' \arrrpar{\Rules}{m}
r'\asub'$.
Since $r'\asub' \arrrpar{\Rules}{k-m-1} \ell\asub$, we have
that $r'\bsub' \arrrpar{\Rules}{k-1} \ell\asub$.
Thus, by the first induction hypothesis, there is $\bsub$ such
that $r'\bsub' \arrr{\Rules} \ell\bsub$ by a formative
$\ell$-reduction, and each $\bsub(\avar) \arrrpar{\Rules}{k-1}
\asub(\avar)$.

We are done if the full reduction $\aterm \arrr{\Rules} \ell'\bsub'
\arr{\Rules} r'\bsub' \arrr{\Rules} \ell\bsub$ is $\ell$-formative;
this is easy with induction on the number of topmost steps in the
second part.
\qed
\end{proof}

Lemma~\ref{lem:base} lays the foundation for all theorems in this
paper.  To start:

\begin{theorem}\label{thm:makechain}
$(\Sigma,\Rules)$ is non-terminating if and only if there is an infinite minimal
formative $(\DP(\Rules),\formrules(\DP(\Rules),\Rules))$-chain.
Here, a chain $\rijtje{(\ell_i \arrz r_i,\asub_i) \mid i \in \N}$ is
\emph{formative} if always $r_i\asub_i \arrr{\formrules(\ell_{i+1},
\Rules)} \ell_{i+1}\asub_{i+1}$ by a formative
$\ell_{i+1}$-reduction.
\end{theorem}

\emph{Proof Sketch:} Construct an infinite $(\DP(\Rules),\Rules)$-chain
following the usual proof, but when choosing $\gamma_{i+1}$, use
Lemma~\ref{lem:base} to guarantee that $r_i\gamma_i
\arrr{\formrules(\ell_{i+1},\Rules)} \ell_{i+1}\gamma_{i+1}$ by a
formative $\ell_{i+1}$-reduction.
\qed

\medskip
Note that this theorem extends the standard dependency pairs result
(Theorem~\ref{thm:dps}) by limiting interest to chains with
formative reductions.

\begin{example}\label{ex:removerun}
The system from Example~\ref{ex:running} is terminating
iff there is no infinite minimal formative $(\Pairs,Q)$-chain, where $\Pairs =
\DP(\Rules)$ from Example~\ref{ex:DPs} and
$Q = \{1,2,3,\linebreak
5,6,7,8,10,11\}$.  Rules $4$ and $9$ have right-hand sides
headed by symbols $\error$ and $\return$ which do not occur
in the left-hand sides of $\DP$ or its formative rules.
\end{example}

\section{Formative Rules in the Dependency Pair Framework}\label{sec:procs}

Theorem~\ref{thm:makechain} provides a basis for using DPs
with formative rules to prove termination: instead of proving
that there is no infinite minimal $(\DP(\Rules),
\Rules)$-chain,
it suffices if there is no infinite minimal formative
$(\DP(\Rules),\formrules(\DP(\Rules),\Rules))$-chain.  
So in
the DP framework, 
we can start with the set
$\{ (\DP(\Rules),\formrules(\DP(\Rules),\Rules), \minimal) \}$
instead of $\{ (\DP(\Rules),\Rules, \minimal) \}$, as we did in  \pagebreak
Example~\ref{ex:removerun}.
We thus
  obtain a similar improvement to Dershowitz'
  refinement~\cite{der:04:1} in that it yields a smaller \emph{initial}
  DP problem: by~\cite{der:04:1}, we can reduce the
  initial set $\DP(\Rules)$; by Theorem~\ref{thm:makechain} we can
  reduce\linebreak the initial set $\Rules$.
However, there (currently) is no way to keep track of the information
that we only need to consider formative chains.  Despite this, we can
define several processors.  All of them are based on 
this
consequence of Lemma~\ref{lem:base}:

\begin{lemma}\label{lem:changechain}
If there is a $(\Pairs,\Rules)$-chain $\rijtje{(\ell_i \arrz r_i,\gamma_i)
\mid i \in \N}$,
then there are $\delta_i$ for $i \in \N$ such that
$\rijtje{(\ell_i \arrz r_i,\delta_i) \mid i \in \N}$ is a formative
$(\Pairs,\formrules(\Pairs,\Rules))$-chain.
\end{lemma}

\begin{proof}
Given $\rijtje{(\ell_i \arrz r_i,\asub_i) \mid i \in \N}$ we construct
the formative chain as follows.
  Let $\bsub_1 := \asub_1$.
  For given $i$, suppose $\bsub_i$ is a substitution such that
  $\bsub_i \arrr{\Rules} \asub_i$, so still $r_i\bsub_i \arrr{\Rules}
  \ell_{i+1}\asub_{i+1}$.  Use Lemma~\ref{lem:base} to find
  $\bsub_{i+1}$ such that $r_i\bsub_i \arrr{\formrules(\ell_{i+1},
  \Rules)} \ell_{i+1}\bsub_{i+1}$
  by a formative $\ell_{i+1}$-reduction,
  and moreover $\bsub_{i+1} \arrr{\Rules} \asub_{i+1}$.
\qed
\end{proof}

This lemma for instance allows us to remove all non-formative rules
from a DP problem.  To this end, we use the following processor:

\begin{theorem}\label{thm:processor}
The DP processor which maps a DP problem
$(\Pairs,\Rules,\flag)$ to the set $\{(\Pairs,  \formrules(\Pairs,
\Rules),\all)\}$ is sound.
\end{theorem}

\noindent
\emph{Proof Sketch:}
This follows immediately from Lemma~\ref{lem:changechain}.
\qed

\begin{example}\label{ex:runningdrop}
Let $Q = \baseformrules(\DP(\Rules),\Rules)$ from Example~\ref{ex:removerun},
and let $\Pairs = \{ \up{\mbig}(\avar,\linebreak
\cons(\bvar,\cvar)) \arrz
\up{\mbig}(\ack(\avar,\bvar),\upd(\cvar))\}$ as in
Example~\ref{ex:runningP}.
If, during a termination proof with dependency pairs, we encounter a
DP problem $(\Pairs,Q,\minimal)$, we can soundly replace it by
$(\Pairs,T,\all)$, where $T = \baseformrules(\Pairs,Q) = \{ 8 \}$.
\end{example}

Thus, we can (permanently) remove all non-formative rules from a
dependency pair problem.  This processor has a clear downside,
however: given a problem $(\Pairs,\Rules,\minimal)$, we lose
minimality.  This $\minimal$ flag is very convenient to have, as
several processors require it (such as reduction pairs with usable
rules from Theorem~\ref{thm:userules}).

Could we preserve minimality?  Unfortunately, the answer
is no.  By modifying a chain to use formative reductions, we may lose
the property that each $r_i\gamma_i$ is terminating.  This happens
for instance for
$(\Pairs,\Rules,\minimal)$, where
$\Pairs = \{ \up{\Fg}(\avar) \arrz \up{\Fh}(\Ff(\avar)),\ 
\up{\Fh}(\Fc) \arrz \up{\Fg}(\Fa) \}$ and $\Rules = \{ \Fa \arrz
\Fb,\ \Ff(\avar) \arrz \Fc,\ \Ff(\Fa) \arrz \Ff(\Fa) \}$. Here,
$\baseformrules(\Pairs,\Rules) = \{ \Ff(\avar) \arrz \Fc,\ 
\Ff(\Fa) \arrz \Ff(\Fa) \}$.
While there is an infinite minimal
$(\Pairs,\Rules)$-chain, the only infinite $(\Pairs,\baseformrules(
\Pairs,\Rules))$-chain is non-minimal.

\medskip
Fortunately, there \emph{is} an easy way to use formative rules
without losing any information: by using them in a reduction pair,
as we typically do for usable rules.  In fact, although usable and
formative rules seem to be opposites, there is no reason why we
should use either one or the other; we can combine them.
Considering also argument filterings, we find the following
extension of Theorem~\ref{thm:userules}.

\begin{theorem}\label{thm:redpairprocessor}
Let $(\succsim,\succ)$ be a reduction pair and $\pi$ an argument
filtering.
The processor which
maps
$(\Pairs,\Rules,\flag)$
to
the following result is sound:
\begin{itemize}
\item $\{(\Pairs \setminus \Pairs^\succ,\Rules,\flag)\}$ if:
  \begin{itemize}
  \item $\filter(\ell) \succ \filter(r)$ for $\ell \arrz r \in \Pairs^\succ$ and
  $\filter(\ell) \succsim \filter(r)$ for $\ell \arrz r \in \Pairs \setminus \Pairs^\succ$;
  \item $u \succsim v$ for $u \arrz v \in \formrules(\filter(\Pairs),\filter(U))$, \\
    where $U = \Rules$ if $\flag = \all$ and $U = \userules(\Pairs,\Rules,\pi) \cup \Ce$ if $\flag = \minimal$;
  \end{itemize}
\item $\{(\Pairs, \Rules, \flag)\}$ otherwise.
\end{itemize}
\end{theorem}\pagebreak

\smallskip\noindent
\emph{Proof Sketch:}
Given an infinite $(\Pairs,\Rules)$-chain, we use argument filterings
and maybe usable rules to obtain a
$(\filter(\Pairs),\filter(U))$-chain which uses the same dependency
pairs infinitely often (as in \cite{gie:thi:sch:fal:06});
using Lemma~\ref{lem:changechain} we turn this chain formative.
\qed

\medskip
Note that we use the argument filtering here in a slightly different
way than for usable rules: rather than including $\pi$ in the
definition of $\formrules$ and requiring that $\filter(\ell) \succsim
\filter(r)$ for $\ell \arrz r \in \formrules(\Pairs,\Rules,\pi)$, we
simply use $\formrules(\filter(\Pairs),\filter(\Rules))$.
\confreport{For space reasons, we}{We} give
additional semantic and syntactic definitions of formative rules with
respect to an argument filtering 
in \confreport{the technical
report~\cite[Appendix~C]{techreportversion}}{Appendix~\ref{sec:wrt}}.

\begin{example}\label{ex:runningredpair}
To handle $(\Pairs,Q,\minimal)$ from Example~\ref{ex:runningdrop}, we
can alternatively use a reduction pair.  Using the trivial
argument filtering, with a polynomial interpretation with $\up{\mbig}
(\avar,\bvar) = \avar + \bvar,\ \ack(\avar,\bvar) = 0,\ 
\upd(\avar) = \avar$ and $\cons(\avar,\bvar) = \bvar + 1$, all
constraints are oriented, and we may remove the only element of
$\Pairs$.

Note that we \emph{could} have handled this example without using
formative rules; $\ack$ and $\random$ can be oriented with an
extension of $\succsim$, or we might use an argument filtering
with $\pi(\up{\mbig}) = \{2\}$.
Both objections could be cancelled by adding extra rules, but we
kept the example short, as it suffices to illustrate the method.
\end{example}

\paragraaf{Discussion}
It is worth noting the parallels between formative and usable rules.
To start, their definitions are very similar; although we did not
present the seman-\linebreak
tic definition of usable rules from~\cite{thi:07}
(which is only used for \emph{innermost} termina\-tion),
the syntactic definitions are almost symmetric.  Also the usage
corresponds: in both cases, we lose minimality when using the direct
rule removing processor, but can safely use the restriction in a
reduction pair (with argument filterings).

There are also differences, however.  The transformations used to
turn a chain usable or formative are very different, with the usable
rules transformation (which we did not discuss) encoding subterms
whose root is not usable, while the formative rules transformation is
simply a matter of postponing reduction steps.

Due to this difference,
usable rules are useful only for a finitely branching system (which
is standard, as all finite MTRSs are finitely branching); formative
rules are useful mostly for left-linear systems (also usual,
especially in MTRSs originating from functional programming, but
typically seen as a larger restriction).
Usable rules introduce the extra $\Ce$ rules, while formative rules
are all included in the original rules.  But for formative rules, even
definitions extending $\baseformrules$, necessarily all collapsing
rules are included, which has no parallel in usable rules;
the parallel of collapsing rules would be rules $\avar
\arrz r$, which are not permitted.

\smallskip
To use formative rules without losing minimality
information, 
an alternative to\linebreak
Theorem~\ref{thm:processor}
allows us to permanently
delete rules.
The trick is to add a new component to 
DP problems, as
for higher-order rewriting in~\cite[Ch.~7]{kop:12}.
A DP problem
 becomes
a tuple $(\Pairs,\Rules,\flag_1,\flag_2)$, with
$\flag_1 \in \{\minimal,
\all\}$ and
$\flag_2 \in \{\formative,\arbitrary\}$,
and is\linebreak finite if there is no infinite $(\Pairs,\Rules)$-chain which is
minimal if $\flag_1 = \minimal$, and formative if $\flag_2 =
\formative$.
By Theorem~\ref{thm:makechain}, $\Rules$ is terminating iff
$(\DP(\Rules),\Rules,\minimal,\formative)$ is finite.

\begin{theorem}\label{thm:flagprocessor}
In the extended 
DP framework, the processor which maps
$(\Pairs,\Rules,\flag_1,\linebreak
\flag_2)$ to $\{(\Pairs,
\formrules(\Pairs,\Rules),\flag_1,\flag_2)\}$ if $\flag_2 =
\formative$ and $\{(\Pairs,\Rules,\flag_1,\flag_2)\}$
otherwise, is sound.
\end{theorem}
\noindent
\emph{Proof:}
This follows immediately from Lemma~\ref{lem:redtorules}.
\qed
\pagebreak

The downside of changing the DP framework in this way
is that we have to revisit all existing DP processors to
see how they interact with the formative flag.
In many cases, we can simply pass the flag on unmodified (i.e. if
$\mathit{proc}((\Pairs,\Rules,\flag_1)) = A$, then $\mathit{proc}'((
\Pairs,\Rules,\flag_1,\flag_2)) = \{ (\Pairs',\Rules',\flag_1',
\flag_2) \mid (\Pairs',\Rules',\flag_1') \in A \}$).  This is for
example the case for processors with reduction pairs (like the one
in Theorem~\ref{thm:redpairprocessor}), the dependency graph and the
subterm criterion.  Other processors would have to be checked
individually, or reset the flag to $\mathsf{arbitrary}$ by default.

Given how long the dependency pair framework has existed (and how
many processors have been defined, see e.g.~\cite{thi:07}), and that
the formative flag clashes with the component for innermost rewriting
(see \secshort~\ref{sec:innermost}), it is unlikely that many tool
programmers will make the effort for a single rule-removing processor.

\section{Handling the Collapsing Rules Problem}
\label{sec:unsorted}

A great weakness
of the formative rules method is the
matter of collapsing rules.\linebreak  Whenever the left-hand side of a
dependency pair or formative rule has a symbol $\afun : [\seq{
\atype}] \decpijl \btype$, 
all collapsing rules of sort
$\btype$ are formative.  And
then all \emph{their} formative rules
are also formative.  
Thus,
this 
often leads to the inclusion of all rules of a given
sort.  In particular for systems with only one sort (such as
all first-order benchmarks
in the Termination Problems Data Base), this is problematic.

For this reason, we will consider a new notion, building on the idea
of formative rules and reductions.  This notion is based on the
observation that it might suffice to include \emph{composite rules}
rather than the formative rules of all collapsing rules.
To illustrate the idea, assume given a uni-sorted system with
rules $\symb{a} \arrz \symb{f}(\symb{b})$ and
$\symb{f}(\avar) \arrz \avar$.
$\baseformrules(\symb{c})$ includes 
$\symb{f}(\avar) \arrz \avar$, 
so
also $\symb{a} \arrz
\symb{f}(\symb{b})$.  But 
a term $\symb{f}(\symb{b})$ does
not reduce to $\symb{c}$.  So intuitively, we should not really need
to include the first rule.

Instead of including the formative rules of all collapsing rules, we
might imagine a system where we \emph{combine} rules with collapsing
rules that could follow them.  In the example above, this gives
$\Rules = \{ \symb{a} \arrz \symb{f}(\symb{b}),\ \symb{a}
\arrz \symb{b},\ \symb{f}(\avar) \arrz \avar \}$.
Now we might consider an alternative definition of formative rules,
where we still need to include the collapsing rule $\symb{f}(\avar)
\arrz \avar$, but no longer need to have $\symb{a} \arrz
\symb{f}(\symb{b})$.

To make this idea formal, we first consider how rules can be combined.
In the following, we consider systems with \emph{only one sort}; this
is needed for the definition to be well-defined, but can always be
achieved by replacing all sorts by $\symb{o}$.

\begin{definition}[Combining Rules]\label{def:combrules}
Given an MTRS $(\Sigma,\Rules)$, let
  $A := \{ \afun(\seq{\avar}) \arrz \avar_i \mid \afun : [\asort_1
  \times \ldots \times \asort_n] \decpijl \bsort
  \in \Sigma \wedge 1 \leq i \leq n \}$ and
  $B := \{ \ell \arrz p \mid \ell \arrz r \in \Rules \wedge r
  \suptermeq p \}$.
Let $X \subseteq A \cup B$ be the smallest set such that
  $\Rules \subseteq X$ and
  for all $\ell \arrz r \in X$:
  \begin{enumerate}[a.]
  \item\label{it:split}
    if $r$ is a variable, $\ell \suptermeq \afun(l_1,\ldots,
    l_n)$ and $l_i \suptermeq r$, then $\afun(\avar_1,\ldots,\avar_n)
    \arrz \avar_i \in X$;
  \item\label{it:combine}
    if $r = \afun(r_1,\ldots,r_n)$ and $\afun(\avar_1,\ldots,
    \avar_n) \arrz \avar_i \in X$, then $\ell \arrz r_i \in X$.
  \end{enumerate}
Let $\collapsing := A \cap X$ and $\noncollapsing = \{ \ell \arrz r
\in X \mid r$ not a variable$\}$.
Let $A_\Rules := \collapsing \cup \noncollapsing$.
\end{definition}

It is easy to see that $\arrr{\Rules}$ is included in
$\arrr{A_\Rules}$: all non-collapsing rules of $\Rules$ are in
$\noncollapsing$, and all collapsing rules are obtained as a
concatenation of steps in $\collapsing$.

\begin{example}\label{ex:AR}
Consider an unsorted version of Example~\ref{ex:running}.  Then for
$(\Pairs,Q)$ as in Example~\ref{ex:runningdrop}, we have $U :=
\userules(\Pairs,Q) = \{1,2,3,5,8,10,11\}$.  Unfortunately, \pagebreak only (3)\linebreak
is not formative, as the two $\random$ rules cause inclusion of all
rules in $\baseformrules(\suc(\avar),U)$.
Let us instead calculate $X$, which we do as an iterative procedure
starting from $\Rules$.  In the following, $C \Rightarrow
D_1,\ldots,D_n$ should be read as:
 ``by requirement \ref{it:split},
rule $C$\linebreak enforces inclusion of each $D_i$'', and $C,D \Rightarrow E$
similarly refers to requirement \ref{it:combine}.
\[
\begin{array}{rclrclrclrclrcl}
2,1 & \Rightarrow & 12 &
5,13 & \Rightarrow & 14 &
10,15 & \Rightarrow & 16 &
16,13 & \Rightarrow & 18 &
17,15 & \Rightarrow & 19 \\
12 & \Rightarrow & 1,13\ \ \ \ &
14 & \Rightarrow & 15\ \ \ \ &
11,15 & \Rightarrow & 17\ \ \ \ &
18 & \Rightarrow & 15,13\ \ \ \ &
19 & \Rightarrow & 15,13 \\
\end{array}
\]
\[
\begin{array}{rrclrrclrrcl}
12. & \random(\suc(\avar)) & \arrz & \avar\ \  &
15. & \ack(\avar,\bvar) & \arrz & \bvar &
18. & \ack(\suc(\avar),\bvar) & \arrz & \bvar \\
13. & \suc(\avar) & \arrz & \avar &
16. & \ack(\suc(\avar),\bvar) & \arrz & \suc(\bvar) &
19. & \ack(\suc(\avar),\suc(\bvar)) & \arrz & \bvar \\
14. & \ack(\nul,\bvar) & \arrz & \bvar &
17. & \ack(\suc(\avar),\suc(\bvar)) & \arrz & \ack(\suc(\avar),\bvar) \\
\end{array}
\]
Now $\collapsing = \{1,13,15\}$ and
$\noncollapsing = \{2,3,5,8,10,11,16,17\}$,
and $A_U = \collapsing \cup \noncollapsing$.
\end{example}

Although combining a system $\Rules$ into $A_\Rules$ may create
significantly more rules, the result is not necessarily harder to
handle. 
For many standard reduction pairs, like RPO or linear polynomials
over $\mathbb{N}$,
we have: if $\aterm
\succsim \avar$ where $\avar \in \FV(\aterm)$ occurs exactly once, then
$\afun(\ldots,\bterm,\ldots) \succsim \bterm$ for any
$\bterm$ with $\aterm \suptermeq \bterm \suptermeq \avar$.
For such
a reduction pair, $A_\Rules$ can be oriented whenever $\Rules$ can be
(if $\Rules$ is left-linear).

$A_\Rules$ has the advantage that we never need to follow a
non-collapsing rule $l \arrz \afun(\seq{r})$ by a collapsing step.
This is essential to use the following definition:

\begin{definition}\label{def:splitform}
Let $A$ be a set of rules.
The \emph{split-formative rules} of a term $\bterm$
are defined as the smallest set $\sformrules(\bterm,A) \subseteq A$
such that:
\vspace{-3pt}
\begin{itemize}
\item if $\bterm$ is not linear, then $\sformrules(\bterm,A) = A$;
\item \fbox{all collapsing rules in $A$ are included in $\sformrules(\bterm,A)$;}
\item if $\bterm = \afun(\bterm_1,\ldots,\bterm_n)$, then
  $\sformrules(\bterm_i,A) \subseteq \sformrules(\bterm,A)$;
\item if $\bterm = \afun(\bterm_1,\ldots,\bterm_n)$, then
  $\{ \ell \to r \in A \mid r$ has the form $\afun(\ldots) \}
  \subseteq \sformrules(\bterm,A)$;
\item if $\ell \arrz r \in \sformrules(\bterm,A)$ \fbox{and $r$ is not a variable}, then
  $\sformrules(\ell,A) \subseteq \sformrules(\bterm,A)$.
\end{itemize}
\vspace{-3pt}
For a set of rules $\Pairs$, we define $\sformrules(\Pairs,A) =
\bigcup_{s \to t \in \Pairs} \sformrules(s,A)$.
\end{definition}

Definition~\ref{def:splitform} is an alternative definition of
formative rules, where collapsing rules have a smaller effect
(differences to Definition~\ref{def:form} are \fbox{highlighted}).
$\sformrules$ is \emph{not} a formative rules approximation, as
shown by the $\symb{a}$-formative reduction
$\symb{f}(\symb{a}) \arr{\Rules} \symb{g}(\symb{a}) \arr{\Rules}
\symb{a}$ with $\Rules = \{\symb{f}(\avar) \arrz \symb{g}(\avar),\ 
\symb{g}(\avar) \arrz \avar \}$
but $\sformrules(\symb{a},\Rules) = \{\symb{g}(\avar) \arrz \avar\}$.
How-\linebreak
ever, given the relation between $\Rules$ and $A_\Rules$, we
find a similar result to Lemma~\ref{lem:redtorules}:

\begin{lemma}\label{lem:redtorulesalternative}
Let $(\Sigma,\Rules)$ be an MTRS.
If $\aterm \arrr{\Rules} \ell\asub$ by a formative $\ell$-reduction,
then $\aterm \arrr{\sformrules(\ell,A_\Rules)} \ell\asub$ by a
formative $\ell$-reduction.
\end{lemma}

Unlike Lemma~\ref{lem:redtorules}, the altered reduction might be
different.  We also do not have that $\sformrules(\Pairs,A_\Rules)
\subseteq \Rules$.  Nevertheless, by this lemma we can use
split-formative rules in reduction pair processors with formative
rules, such as Theorem~\ref{thm:redpairprocessor}.

\medskip\noindent
\emph{Proof Sketch:}
The original reduction $\aterm \arrr{\Rules} \ell\asub$ gives rise to
a formative reduction over $A_\Rules$, simply replacing collapsing
steps by a sequence of rules in $\collapsing$.
So, we assume given a formative $\ell$-reduction over $A_\Rules$, and
prove with induction first on the number of non-collapsing steps in
the reduction, second on the length of the reduction, third on the
size of $\aterm$, that $\aterm \arrr{\sformrules(\ell,A_\Rules)}
\ell\asub$ by a formative $\ell$-reduction. 

This is mostly easy with the induction hypotheses; note that if a
root-rule in $\noncollapsing$ is followed by a rule in $\collapsing$,
there can be no internal $\arrr{\Rules}$ reduction \pagebreak in between (as
this would not be a formative reduction); combining a rule in
$\noncollapsing$ with a rule in $\collapsing$ gives either a rule in
$\noncollapsing$ (and a continuation with the second induction
hypothesis) or a sequence of rules in $\collapsing$ (and the first
induction hypothesis). \qed

\medskip
Note that this method unfortunately does not transpose directly to
the higher-order setting, where collapsing rules may have more
complex forms.
We also had to give up sort differentiation, as otherwise we might
not be able to flatten a rule $\afun(\bfun(\avar)) \arrz \avar$ into
$\afun(\avar) \arrz \avar,\ \bfun(\avar) \arrz \avar$.  This is not
such a great problem, as reduction pairs typically do not care about
sorts, and we circumvented the main reason why sorts are important
for formative rules.
We have the following result:

\begin{theorem}
\label{thm:altfilter}
Let $(\succsim,\succ)$ be a reduction pair and $\pi$ an argument
filtering.
The processor which maps a DP problem $(\Pairs,\Rules,\flag)$
to the following result is
sound:
\begin{itemize}
\item $\{(\Pairs \setminus \Pairs^\succ,\Rules,\flag)\}$ if:
  \begin{itemize}
  \item $\filter(\ell) \succ \filter(r)$ for $\ell \arrz r \in \Pairs^\succ$ and
  $\filter(\ell) \succsim \filter(r)$ for $\ell \arrz r \in \Pairs \setminus \Pairs^\succ$;
  \item $u \succsim v$ for $u \arrz v \in \sformrules(\filter(\Pairs),A_{\filter(U)})$
    and $\FV(t) \subseteq \FV(s)$ for $s \arrz t \in \filter(U)$,
    where $U = \Rules$ if $\flag = \all$ and $U =
    \userules(\Pairs,\Rules,\pi)
    \cup \Ce$ if $\flag = \minimal$;
  \end{itemize}
\item $\{(\Pairs, \Rules, \flag)\}$ otherwise.
\end{itemize}
\end{theorem}

\noindent
\emph{Proof Sketch:}
Like Theorem~\ref{thm:redpairprocessor}, but using
Lemma~\ref{lem:redtorulesalternative} to alter the created
formative $(\filter(\Pairs),\filter(U))$-chain to a split-formative
$(\filter(\Pairs),\sformrules(\filter(\Pairs),A_{\filter(U)}))$-chain.
\qed

\begin{example}
Following Example~\ref{ex:AR}, $\sformrules(\up{\mbig}(\avar,
\cons(\bvar,\cvar)) \arrz \up{\mbig}(\ack(\avar,\bvar),\linebreak
\upd(\cvar)),A_U) = \collapsing \cup \{8\}$, and 
Theorem~\ref{thm:altfilter}
gives an easily orientable problem.
\end{example}

\section{Formative Rules for Innermost Termination}\label{sec:innermost}

So far, we have considered only \emph{full termination}.  A very
common related query is \emph{innermost termination}; that is,
termination of $\arr{\Rules}^{\mathsf{in}}$, defined by:
\begin{itemize}
\item $\afun(\seq{l})\gamma \arrin{\Rules} r\gamma$ if $\afun(\seq{l})
  \arrz r \in \Rules$,\ $\gamma$ a substitution and all
  $l_i\gamma$ in normal form;
\item $\afun(\aterm_1,\ldots,\aterm_i,\ldots,\aterm_n) \arrin{\Rules}
  \afun(\aterm_1,\ldots,\aterm_i',\ldots,\aterm_n)$ if
  $\aterm_i \arrin{\Rules} \aterm_i'$.
\end{itemize}
The innermost reduction relation
is often used in for instance program analysis.

An innermost strategy can be included in the dependency pair framework
by adding 
the $\innermost$ flag \cite{gie:thi:sch:fal:06} to DP problems
(or, more generally, a component $\QQ$ \cite{gie:thi:sch:05:2}
which indicates that when reducing
any term with $\arr{\Pairs}$ or $\arr{\Rules}$, its strict subterms
must be normal with respect to $\QQ$).
Usable rules are more viable for innermost than normal termination:
we do not need minimality,
the $\Ce$ rules do not need to be 
handled by the reduction pair, and we
can define a sound 
processor that maps $(\Pairs,\Rules,
\flag,\innermost)$ to $\{(\Pairs,\userules(\Pairs,\Rules),\flag,
\innermost)\}$.

This is not the case for formative rules.
Innermost
reductions 
eagerly evaluate arguments, 
yet formative reductions
postpone 
evaluations
as long as possible.  In a way, these
are exact opposites.
 Thus, it should not be
surprising that
formative rules are \emph{weaker} for innermost termination than for
full termination. 
  Theorem~\ref{thm:makechain} has no counterpart 
  for $\arr{\Rules}^{\mathsf{in}}$;
  for
  innermost termination
  we must start the 
  DP
  framework with
  $(\Pairs,\Rules,\minimal,\innermost)$, not with
  $(\Pairs,\formrules(\Pairs,\Rules),\minimal,\innermost)$.
  Theorem~\ref{thm:processor} is only sound if the innermost flag
  is removed: $(\Pairs,\Rules,\flag,\innermost)$ is
  mapped to $\{(\Pairs,\formrules(\Pairs,\Rules),\all,\arbitrary)\}$.
Still,
we \emph{can} safely use formative rules with
reduction pairs.
For example, we obtain this variation of
Theorem~\ref{thm:redpairprocessor}: \pagebreak

\begin{theorem}\label{thm:innermostredpair}
Let $(\succsim,\succ)$ be a reduction pair and $\pi$ an argument
filtering.
The processor which 
maps a DP problem $(\Pairs,\Rules,\flag_1,
\flag_2)$
to
the following result is sound:
\begin{itemize}
\item $\{(\Pairs \setminus \Pairs^\succ,\Rules,\flag_1,\flag_2)\}$ if:
  \begin{itemize}
  \item $\filter(\ell) \succ \filter(r)$ for $\ell \arrz r \in \Pairs^\succ$ and
  $\filter(\ell) \succsim \filter(r)$ for $\ell \arrz r \in \Pairs \setminus \Pairs^\succ$;
  \item $u \succsim v$ for $u \arrz v \in \formrules(\filter(\Pairs),\filter(U))$, where $U$ is:
      $\userules(\Pairs,\Rules,\pi)$ if $\flag_2 = \innermost$;
      otherwise
        $\userules(\Pairs,\Rules,\pi) \cup \Ce$ if 
        $\flag_1 = \minimal$;
      otherwise
      $\Rules$.
  \end{itemize}
\item $\{(\Pairs, \Rules, \flag_1, \flag_2)\}$ otherwise.
\end{itemize}
\end{theorem}

\noindent
\emph{Proof Sketch:}
The proof of Theorem~\ref{thm:redpairprocessor} still applies; we
just ignore that the given chain might be innermost (aside from
getting more convenient usable rules).
\qed

\medskip
Theorem~\ref{thm:altfilter} extends to innermost termination in a
similar way.

Conveniently,
innermost
termination is \emph{persistent} \cite{fuh:gie:par:sch:swi:11}, so
modifying $\Sigma$  does not alter innermost
termination behaviour, as long as all rules 
stay
well-sorted.
In practice, we could infer a typing with as many different sorts
as possible, and get stronger formative-rules-with-reduction-pair
processors.
With
the \emph{innermost switch processor}~\cite[Thm.~3.14]{thi:07},
which 
in cases 
can set the $\mathsf{innermost}$ flag
on a DP problem, we could also often use this trick even 
for proving
full termination.

In \secshort~\ref{sec:procs}, we used the
extra flag $\flag_2$ as the formative flag.  It 
is not
contradictory to use $\flag_2$ 
in both ways, allowing $\flag_2 \in
\{ \arbitrary, \formative, \innermost \}$,
since it is\linebreak very unlikely for a $(\Pairs,\Rules)$-chain to be both
formative and innermost at once!  When\linebreak  using both extensions of the
DP framework together, termination provers (human or
computer) will, however, sometimes have to make a choice which flag
to add.

\section{Implementation and Experiments}\label{sec:code}

We have performed a preliminary implementation of formative rules in
the termination tool \aprove\ \cite{gie:bro:emm:fro:fuh:ott:plu:sch:str:swi:thi:14}.
Our automation borrows from the usable rules
of~\cite{gie:thi:sch:05:1}\linebreak (see
\confreport{\cite[Appendices B+D]{techreportversion}}{Appendices
  \ref{app:tcap} + \ref{app:aprove_impl}})
and uses a constraint encoding
\cite{cod:gie:sch:thi:12} 
for 
a
combined search for
argument filterings and corresponding formative rules.
While we did not find\linebreak any termination proofs for examples from
the TPDB where none were known before, 
our experiments show
that formative
rules do improve the power of reduction pairs 
for widely used
term orders (e.g., polynomial orders~\cite{lan:79}).
For 
more information,
 see also:
\url{http://aprove.informatik.rwth-aachen.de/eval/Formative}

For instance, we experimented with a configuration where we applied
dependency pairs,
and then alternatingly dependency graph
decomposition and
reduction pairs
with linear polynomials
and coefficients $\leq 3$.
On the TRS Standard category
of the TPDB (v8.0.7) with 1493 examples,
this configuration (without formative rules, but with usable rules
w.r.t.\ an argument filter) shows termination of 
579 examples within a timeout of 60 seconds (on an 
Intel Xeon 5140 at 2.33 GHz). 
With additional formative
rules, our implementation of Theorem~\ref{thm:redpairprocessor}
proved termination of 6 additional
TRSs. 
(We did, however, lose 4 examples to
timeouts, which we believe are due in part to the
currently unoptimised implementation.
)

The split-formative rules from Theorem \ref{thm:altfilter} are not a
subset of $\Rules$,  in contrast to\linebreak the usable rules. Thus, it is
\emph{a priori} not clear how  to combine their
encodings  
w.r.t.\ an argument filtering, and we conducted experiments 
using only the standard usable rules.
Without formative rules, \pagebreak
532 examples are proved terminating. In contrast, adding either the
formative rules of Theorem \ref{thm:redpairprocessor} or
the split-formative rules of Theorem \ref{thm:altfilter}
we solved 6 additional examples each (where Theorem
\ref{thm:redpairprocessor} and Theorem \ref{thm:altfilter}
each had 1 example the other could not solve), 
losing 1 to timeouts.

Finally, we experimented with the improved dependency pair
transformation based on Theorem \ref{thm:makechain},
which drops non-formative rules from $\Rules$.
We applied DPs as the first technique on the 1403 TRSs
from TRS Standard with
at least one DP. This 
reduced the number of rules in the initial DP problem
for 618 of these TRSs, without 
any search problems and
without sacrificing minimality.

Thus, our current impression is that while formative rules are not the
next ``killer technique'', they nonetheless provide additional
power to widely-used orders in an elegant way and reduce the number of
term constraints to be solved in a termination proof.
The examples from the TPDB are all untyped, and we believe
that formative rules may have a greater impact in a typed first-order setting.

\section{Conclusions}\label{sec:future}

In this paper, we have simplified the notion of formative rules
from~\cite{kop:raa:12:1} to the first-order setting, and integrated
it in the \confreport{dependency pair}{DP} framework.  We did so by means of
\emph{formative reductions}, which allows us to obtain a semantic
definition of formative rules (more extensive syntactic definitions
are discussed in~\confreport{\cite{techreportversion}}{the appendix}).

We have defined three processors to use formative rules in the
standard dependency pair framework for full termination: one is a
processor to permanently remove rules, the other two combine
formative rules with a reduction pair.

We also discussed how to strengthen the method by adding a new
flag to the framework -- although doing so might require too many
changes to existing processors and strategies to be considered
worthwhile -- and how we can still use the technique in the
innermost case, and even profit from the innermost setting.

\vspace{-1ex}
\paragraaf{Related Work}
In the first-order 
DP
framework 
two processors
stand out as relevant to formative rules.  The first is, of
course, usable rules;
see \secshort~\ref{sec:procs} for a detailed discussion.
The
second is the \emph{dependency graph}, which determines whether any
two dependency pairs can follow each other in a
$(\Pairs,\Rules)$-chain, and uses this information to eliminate
elements of $\Pairs$, or to split $\Pairs$ in multiple parts.

In state-of-the-art implementations of the dependency graph
(see e.g.~\cite{thi:07}), both left- and right-hand side of
dependency pairs are considered to see whether a pair can be preceded
or followed by another pair.  Therefore it seems quite 
surprising
that the same mirroring was not previously tried for usable rules.

Formative rules \emph{have} been previously defined, for higher-order
term rewriting, in~\cite{kop:raa:12:1}, which introduces a limited
DP framework, with formative rules (but not formative reductions)
included in the definition of a chain:  we simply impose the
restriction    that always $r_i\gamma_i \arrr{\formrules(\Pairs,
\Rules)} \ell_{i+1}\gamma_{i+1}$.  This gives a reduction pair
processor 
which considers only formative rules, although it cannot
be combined with usable rules and argument filterings.
The authors do not yet consider rule removing processors, but if they
did, Theorem~\ref{thm:flagprocessor} would also go through.

In the second author's PhD thesis~\cite{kop:12}, a more complete
higher-order DP framework is considered. Here, we
do\confreport{\pagebreak}{}\ see formative reductions, and a variation of Lemma~\ref{lem:base}
which, however, requires that $\aterm$ is terminating: the proof
style used here does not go through there due to $\beta$-reduction.
Consequently, Lemma~\ref{lem:changechain} does not go through in the
higher-order setting, and there is no counterpart to
Theorems~\ref{thm:processor} or~\ref{thm:redpairprocessor}.  We
\emph{do}, however, have Theorem~\ref{thm:flagprocessor}.
Furthermore, the results of \secshort~\ref{sec:unsorted} are entirely
new to this paper, and do not apply in the higher-order setting,
where 
rules might also have a form $l \arrz \avar \cdot
\aterm_1 \cdots \aterm_n$ (with $\avar$ a variable).

\vspace{-1ex}
\paragraaf{Future Work}
In the future, it would be interesting to look back at higher-order
rewriting, and see whether we can obtain some form of
Lemma~\ref{lem:changechain} after all.  Alternatively, we might be
able to use the specific form of formative chains to obtain formative
(and usable) rules w.r.t.\ an argument filtering.

In the first-order setting, we might turn our attention to
non-left-linear rules.
Here, we could
 think
 for instance
 of renaming
 apart some of these variables; a rule $\afun(\avar,
\avar) \arrz \bfun(\avar,\avar)$ could 
become any of
$\afun(\avar,\bvar) \arrz \bfun(\avar,\bvar),\ 
\afun(\avar,\bvar) \arrz \bfun(\bvar,\avar), \;\ldots$

\bibliographystyle{plain}
\bibliography{references}

\confreport{}{
\newpage

\appendix

\newcommand{\addtop}{\mathsf{addtop}}
\newcommand{\encode}{\mathcal{I}}
\newcommand{\makelist}{\mathsf{makelist}}
\newcommand{\afilterleft}{\filter_{\mathsf{var}}}
\newcommand{\afilterright}{\filter_\top}
\newcommand{\afilter}{\filter_\top}

\section{Omitted Proofs}
\label{app:proofs}

\subsection{Conventions and Basic Proofs}\label{app:basics}

In this appendix, we will use the following conventions:
\begin{itemize}
\item If $A$ is a function mapping terms to terms and $\asub$ is a
  substitution, then $\asub^A$ indicates the substitution on the same
  domain as $\asub$ such that always $\asub^A(\avar) =
  A(\asub(\avar))$.
\item A relation $\arr{\Rules}^=$ indicates $0$ or $1$ steps with
  $\arr{\Rules}$ (this is a sub-relation of $\arrr{\Rules}$).
\item Whenever we write $\pi(\afun) = \{i_1,\ldots,i_k\}$, we
  implicitly mean that $i_1 < \ldots < i_k$.
\item If $\pi(\afun) = \{1,\ldots,\mathit{arity}(\afun)\}$, then
  we identify $\afun_\pi$ and $\afun$.  Here, $\mathit{arity}$
  refers to the number of arguments a function symbol might take:
  if $\afun : [\asort_1 \times \ldots \times \asort_n] \decpijl
  \bsort \in \Sigma$ then $\mathit{arity}(\afun) = n$.
\end{itemize}

The following results about filtering will be useful on several
occasions, as filtering plays an important role in several of our
theorems.

\begin{lemma}\label{lem:filtersubstitute}
For all terms $\aterm$ and substitutions $\gamma$:
$\filter(\aterm\gamma) = \filter(\aterm)\gamma^{\filter}$
\end{lemma}

\begin{proof}
By induction on $\aterm$.

If $\aterm$ is a variable, then $\filter(\aterm\gamma) = \filter(
\gamma(\aterm)) = \gamma^{\filter}(\aterm) = \filter(\aterm)\gamma^{
\filter}$.

If $\aterm = \afun(\aterm_1,\ldots,\aterm_n)$ with $\pi(\afun) =
\{i_1,\ldots,i_k\}$, then by the induction hypothesis:
\[
\begin{array}{rcl}
\filter(\aterm\gamma) & = &
\filter(\afun(\aterm_1\gamma,\ldots,\aterm_n\gamma)) \\
& = & \afun_\pi(\filter(\aterm_{i_1}\gamma),\ldots,
  \filter(\aterm_{i_k}\gamma)) \\
& = & \afun_\pi(\filter(\aterm_{i_1})\gamma^{\filter},\ldots,
  \filter(\aterm_{i_k})\gamma^{\filter}) \\
& = & \afun_\pi(\filter(\aterm_{i_1}),\ldots,\filter(\aterm_{i_k}))
  \gamma^{\filter} \\
& = & \filter(\aterm)\gamma^{\filter} \\
\end{array}
\]
\qed
\end{proof}

\begin{lemma}\label{lem:filterreduce}
For all terms $\aterm,\bterm$: if $\aterm \arr{\Rules} \bterm$ then
$\filter(\aterm) \arr{\filter(\Rules)}^= \filter(\bterm)$.
\end{lemma}

Note that this lemma holds for any set $\Rules$, so also for a set
$\Rules = \{ \ell \arrz r \}$.  Thus, we know that a filtered version
of the same rule is used.

\begin{proof}
By induction on the derivation of $\aterm \arr{\Rules} \bterm$.

If $\aterm = \afun(\aterm_1,\ldots,\aterm_j,\ldots,\aterm_n)$ and
$\bterm = \afun(\aterm_1,\ldots,\bterm_j,\ldots,\aterm_n)$ with
$\aterm_j \arr{\Rules} \bterm_j$ and
$\pi(\afun) = \{i_1,\ldots,i_k\}$,
then we consider two possibilities.
If $j = i_l$ for some $l$, then
$\filter(\aterm) = \afun_\pi(\filter(\aterm_{i_1}),\ldots,
\filter(\aterm_{i_l}),\ldots,\filter(\aterm_{i_k})) \arr{\filter(
\Rules)}^= \afun_\pi(\filter(\aterm_{i_1}),\ldots,\filter(\bterm_{
i_l}),\ldots,\filter(\aterm_{i_k}))$ by the induction hypothesis,
$= \afun_\pi(\bterm)$.
Otherwise, if $j \notin \pi(\afun)$, then
$\filter(\aterm) = \filter(\bterm)$.

Alternatively, if the reduction takes place at the root, we can write
$\aterm = \ell\gamma,\ \bterm = r\gamma$ for some $\ell \arrz r \in
\Rules$.  Then $\filter(\aterm) = \filter(\ell)\gamma^{\filter}$ by
Lemma~\ref{lem:filtersubstitute}, $\arr{\filter(\Rules)} \filter(r)
\gamma^{\filter}$ because $\filter(\ell) \arrz \filter(r) \in
\filter(\Rules)$, and this is exactly $\filter(\bterm)$ by
Lemma~\ref{lem:filtersubstitute} again.
\qed
\end{proof}

\subsection{Preliminaries Proofs}\label{app:preliminaries}

Theorems~\ref{thm:dps} and~\ref{thm:redpair} are presented for
understanding, but are not used in the proofs of other results in this
paper.  As they are moreover an immediate consequence of
Theorems~\ref{thm:makechain} and~\ref{thm:redpairprocessor} (which are
proved in much the same way in Section~\ref{app:mainproofs}), we do
not prove them here, but simply refer
to~\cite{art:gie:00:1,gie:thi:sch:05:2,gie:thi:sch:fal:06,hir:mid:07:1}.

Theorem~\ref{thm:userules} is also a standard result from the
literature \cite{gie:thi:sch:fal:06,hir:mid:07:1},
but this result (and even more important: the main lemma
that leads to it) \emph{is} used in other proofs.  Moreover, seeing
the proof will aid understanding for the not-entirely-standard proof
in Appendix~\ref{sec:wrt}.
To start, we will need some definitions and lemmas.

\begin{definition}
Fixing an argument filtering $\pi$ and sets of rules $\Pairs$ and
$\Rules$ such that $\Rules$ is finitely branching, we introduce for any
sort $\asort$ fresh symbols $\bot_\asort : \asort$ and
$\Fc_\asort : [\asort \times \asort] \decpijl \asort$.
The function $\encode$ from \emph{terminating}
terms to terms is inductively defined by:
\begin{itemize}
\item if $\aterm$ is a variable, then $\encode(\aterm) = \aterm$;
\item if $\aterm$ has the form $\afun(\aterm_1,\ldots,\aterm_n)$,
  with $\afun(\seq{\aterm}) : \asort$ and $\pi(\afun) = \{i_1,
  \ldots,i_k\}$, then:
  \begin{itemize}
  \item if $\afun$ is a \emph{usable symbol}**, then
    $\encode(\aterm) = \afun_\pi(\encode(\aterm_{i_1}),\ldots,
    \encode(\aterm_{i_k}))$
  \item otherwise,
    $\encode(\aterm) = \Fc_\asort(\afun_\pi(\encode(\aterm_{i_1}),\ldots,
    \encode(\aterm_{i_k})),\makelist_\asort(\{\bterm \mid
    \afun(\seq{\aterm}) \arr{\Rules} \bterm\}))$, where
    $\makelist_\asort$ is defined by:
    \begin{itemize}
    \item $\makelist_\asort(\emptyset) = \bot_\asort$
    \item $\makelist_\asort(X) = \Fc_\asort(\encode(\bterm),
      \makelist_\asort(X \setminus \{\bterm\}))$ if $X$ is non-empty
      and $\bterm$ is its smallest element (lexicographically).
    \end{itemize}
  \end{itemize}
\end{itemize}
** We say $\afun$ is a usable symbol if either $\afun$ is a
constructor, or $\afun$ is the root symbol of some rule in
$\filteruserules(\Pairs,\Rules,\pi)$.
\end{definition}

This definition is well-defined because the system is finitely
branching (so $X$ is always finite) and terminating.

\begin{definition}
We say a term is \emph{completely $\pi$-usable} if either it is a
variable, or it has the form $\afun(\aterm_1,\ldots,\aterm_n)$ with
$\afun$ a usable symbol and for all $i \in \pi(\afun)$ also
$\aterm_i$ is completely $\pi$-usable.
\end{definition}

This definition particularly applies to the right-hand sides of
dependency pairs and usable rules which, by definition, are always
completely $\pi$-usable.

\begin{lemma}\label{lem:ur1}
For all substitutions $\gamma$ and terms $\aterm$ we have
$\encode(\aterm\gamma) \arrr{\Ce} \filter(\aterm)\gamma^\encode$.
If $\aterm$ is completely $\pi$-usable, then even
$\encode(\aterm\gamma) = \filter(\aterm)\gamma^\encode$.
\end{lemma}

\begin{proof}
By induction on the size of $\aterm$.

For a variable we easily see that $\encode(\aterm\gamma) =
\encode(\gamma(\aterm)) = \gamma^\encode(\aterm) =
\filter(\aterm)\gamma^\encode$.

If $\aterm = \afun(\aterm_1,\ldots,\aterm_n)$, let $\pi(\afun) =
\{i_1,\ldots,i_k\}$.  We either have that $\encode(\aterm\gamma) =
\afun_\pi(\encode(\aterm_{i_1}\gamma),
\ldots,\encode(\aterm_{i_k}\gamma))$, or
$\encode(\aterm\gamma)$ reduces to this term with a single $\Ce$ step;
the latter case cannot occur if $\aterm$ is completely $\pi$-usable, as
then $\afun$ would be a usable symbol.
Either way, as $\filter(\aterm)\gamma^\encode =
\afun_\pi(\filter(\aterm_{i_1})\gamma^\encode,\ldots
,\filter(\aterm_{i_k})\gamma^\encode)$ and all $\aterm_{i_j}$ are
completely $\pi$-usable if $\aterm$ is, the induction hypothesis
gives the desired result.
\qed
\end{proof}

\begin{lemma}\label{lem:ur2}
If $\aterm$ is a terminating term and $\Rules$ finitely branching, and
$\aterm \arr{\Rules} \bterm$,
then $\encode(\aterm) \arrr{\filter(\filteruserules(\Pairs,\Rules,\pi))
\cup \Ce} \encode(\bterm)$.
\end{lemma}

\begin{proof}
By induction on the derivation of $\aterm \arr{\Rules} \bterm$. 
Since $\aterm$ reduces, it cannot be a variable, so let us write
$\aterm = \afun(\aterm_1,\ldots,\aterm_n)$, with $\pi(\afun) =
\{i_1,\ldots,i_k\}$.  For brevity, also denote
$U := \filter(\filteruserules(\Pairs,\Rules,\pi))$.

If $\afun$ is not a usable symbol, then $\encode(\aterm) \arr{\Ce}
\makelist_\asort(X)$, where $X$ is a set which contains $\bterm$.
It is easy to see that therefore $\makelist_\asort(X) \arrr{\Ce}
\encode(\bterm)$.  Thus, let us henceforth assume that $\afun$ is a
usable symbol.

If
$\afun(\aterm_1,\ldots,
\aterm_n) = \ell\gamma$ and $\bterm = r\gamma$ for some rule $\ell
\arrz r \in \Rules$ and substitution $\gamma$.  We can write $\ell =
\afun(l_1,\ldots,l_n)$, and since $\afun$ is a usable symbol,
$\filter(\ell) \arrz \filter(r) \in U$.  By Lemma~\ref{lem:ur1},
$\encode(\aterm)
\arrr{\Ce} \filter(\ell)\gamma^\encode \arr{U}
\filter(r)\gamma^\encode = \encode(r\gamma) =
\encode(\bterm)$ as required, because the right-hand sides of usable
rules are always completely $\pi$-usable.

Alternatively, $\bterm = \afun(\aterm_1,\ldots,\aterm_i',\ldots,
\aterm_n)$ with $\aterm_i \arr{\Rules} \aterm_i'$.  If $i \notin
\pi(\afun)$, then $\encode(\aterm) = \encode(\bterm)$.  Otherwise,
by the induction hypothesis:
$\encode(\aterm_i) \arrr{U \cup \Ce} \encode(\aterm_i')$.
\qed
\end{proof}

\begin{lemma}\label{lem:usablechain}
If there is an infinite minimal $(\Pairs,\Rules)$-chain
$\rijtje{(\ell_i \arrz r_i,\gamma_i) \mid i \in \N}$, then there are
$\delta_1,\delta_2,\ldots$ such that $\rijtje{(\filter(\ell_i) \arrz
\filter(r_i),\delta_i) \mid i \in \N}$ is an infinite
$(\filter(\Pairs),\filter(\filteruserules(\Pairs,\Rules,\pi)) \cup
\Ce)$-chain.
\end{lemma}

Note that, as long as we let $\pi(\Fc_\asort) = \{1,2\}$ for all
$\asort$, we have $\filter(\filteruserules(\Pairs,\Rules,\pi)) \cup
\Ce = \filter(\filteruserules(\Pairs,\Rules,\pi) \cup \Ce)$.  We will
typically use the two interchangeably, and simply assume that the
$\Fc_\asort$ symbols are not filtered.

\begin{proof}
We can safely assume that the domains of all $\gamma_i$ contain only
variables in $\FV(\ell_i) \cup \FV(r_i)$; otherwise, we limit the
domain as such.

The lemma is obvious if $\Rules$ is not finitely branching; otherwise
$\encode$ is defined on all terminating terms, such as the strict
subterms of all $r_i\gamma_i$ and $\ell_i\gamma_i$.  This includes
all $\gamma_i(\avar)$ by the assumption.
Thus, for all $i$, we define $\delta_i = \gamma_i^\encode$.  Then we
have: $r_i\gamma \arrr{\Rules} \ell_{i+1}\gamma$ implies that
$\filter(r_i)\delta_i =
\encode(r_i\gamma_i)$ (by Lemma~\ref{lem:ur1}) $\arrr{\filter(
\filteruserules(\Pairs,\Rules,\pi)) \cup \Ce}
\encode(l_i\gamma_{i+1})$ (by Lemma~\ref{lem:ur2}) $\arrr{\Ce}
\filter(l_i)\delta_{i+1}$ (by Lemma~\ref{lem:ur1}).
\qed
\end{proof}

\begin{proof}[Proof of Theorem~\ref{thm:userules}]
Suppose there is an infinite minimal $(\Pairs
\setminus \Pairs^\succ,\Rules)$-chain $\rijtje{(\ell_i \arrz r_i,
\gamma_i) \mid i \in \N}$.  By Lemma~\ref{lem:usablechain}, we can
find $\delta_1,\delta_2,\ldots$ such that $\rijtje{(\filter(\ell_i)
\arrz \filter(r_i),\delta_i) \mid i \in \N}$ is a chain.  This new
chain uses (filtered versions of) the same dependency pairs infinitely
often, but uses $\filter(\filteruserules(\Pairs,\Rules,\pi)) \cup \Ce$
for the reduction $\filter(r_i)\delta_i \arrz^* \filter(\ell_{i+1})
\delta_{i+1}$.  The requirements guarantee that the new chain
cannot use filtered elements of $\Pairs^\succ$ infinitely often.
Hence, the original chain must have a tail without these pairs.
\qed
\end{proof}

\subsection{Proofs for Sections~\ref{sec:formrules}--\ref{sec:unsorted}}\label{app:mainproofs}

\begin{proof}[Proof of Theorem~\ref{thm:makechain}]
To handle the \emph{if} side, we refer to~\cite{art:gie:00:1}: if
there is an infinite minimal formative $(\DP(\Rules),\formrules(
\DP(\Rules),\Rules))$-chain, then there is an infinite $(\DP(\Rules),
\Rules)$-chain, so $(\Sigma,\Rules)$ is non-terminating.
For the \emph{only if} side,
let $(\Sigma,\Rules)$ be non-terminating; we must find
an infinite minimal formative $(\DP(\Rules),\formrules(\DP(\Rules),
\Rules))$-chain.  Note that by definition of an MTRS, the rules in
$\Rules$ are required to satisfy the variable condition: always
$\FV(r) \subseteq \FV(\ell)$.

Assume $\Rules$ is non-terminating, so there is a non-terminating
term $\aterm_{-1}$; without loss of generality we can assume that all
subterms of $\aterm_{-1}$ do terminate.

For $i \in \mathbb{N} \cup \{-1\}$, let a minimal non-terminating
term $\aterm_i$ be given.
By minimality, we know that any infinite reduction starting in
$\aterm_i$ eventually has to take a root step, say $\aterm_i =
\afun(\bterm_1,\ldots,\bterm_n) \arrr{\Rules} \afun(\cterm_1,
\ldots,\cterm_n) = \ell\gamma \arr{\Rules} r\gamma$ where $\ell
\arrz r$ is a rule and $r\gamma$ is still non-terminating.

Write $\ell = \afun(l_1,\ldots,l_n)$.  If $\ell$ is non-linear,
a reduction $\up{\afun}(\seq{\bterm}) \arrr{\Rules} \up{\afun}(
\seq{l})\gamma$ is $\ell$-formative; in this case let $\delta :=
\gamma$.
If $\ell$ is linear, then by Lemma~\ref{lem:base}, we can find
$\delta_1,\ldots,\delta_n$ on the disjoint domains $\FV(l_1),
\ldots,\FV(l_n)$ such that each $\bterm_i \arrr{\formrules(l_i,
\Rules)} l_i\delta_i$ by a formative $l_i$-reduction, and each
$\delta_i(\avar) \arrr{\Rules} \gamma(\avar)$; we define $\delta :=
\delta_1 \cup \ldots \cup \delta_n$ and note that $r\delta
\arrr{\Rules} r\gamma$ is also non-terminating.

Thus, either way we have a substitution $\delta$ such that
$\aterm_i \arrr{\formrules(\ell,\Rules)} \up{\afun}(l_1,\ldots,
l_n)$
by a formative $\ell$-reduction, and $r\delta$ is still
non-terminating.  Let $p$ be a minimal subterm of $r$ such that
$p\delta$ is still non-terminating.  Then $p$ cannot be a variable,
as then $p\delta$ would be a subterm of some $l_j\delta$,
contradicting termination of $\bterm_j$.  Thus, $p = \bfun(p_1,
\ldots,p_m)$.  If $\bfun$ is not a defined symbol, then any infinite
reduction starting in $p\delta$ would give an infinite reduction in
some $p_i\delta$, contradicting minimality; so, $\bfun$ is defined.
We conclude that $\ell_{i+1} := \up{\afun}(\seq{l}) \arrz
\up{\bfun}(\seq{p}) =: r_{i+1}$ is a dependency pair.  Let
$\gamma_{i+1} := \delta$ and $\aterm_{i+1} := p\delta$.

This builds an infinite $(\DP(\Rules),\Rules)$-chain with
all $\arr{\Rules}$-steps formative!
\qed
\end{proof}

\begin{proof}[Proof of Theorem~\ref{thm:processor}]
Suppose $(\Pairs,\Rules,\flag)$ is not finite, so
there is an infinite $(\Pairs,\Rules)$-chain $\rijtje{(\ell_i,r_i,
\asub_i) \mid i \in \N}$ with all $\ell_i \arrz r_i \in \Pairs$ and
always $r_i\asub_i \arrr{\Rules} \ell_{i+1}\asub_{i+1}$.
By Lemma~\ref{lem:changechain} there is also an infinite $(\Pairs,
\formrules(\Pairs,\Rules))$-chain.
\qed
\end{proof}

\begin{proof}[Proof of Theorem~\ref{thm:redpairprocessor}]
We can safely assume that $\pi(\Fc_\asort) = \{1,2\}$, as the
$\Fc_\asort$ symbols occur only in $\Ce$: if we can prove the theorem
for such an argument filtering, then for general filterings $\rho$ we
write $\rho = \sigma \circ \pi$, where $\pi$ is the restriction of
$\rho$ to symbols other than the $\Fc_\asort$ (with additional entries
$\pi(\Fc_\asort) = \{1,2\}$) and $\sigma$ maps $\Fc_\asort$ to
$\rho(\Fc_\asort)$ and other $\afun_\pi$ to $\{1,\ldots,k\}$ if
$\pi(\afun) = \{i_1,\ldots,i_k\}$.  Then, if the processor is
applicable with argument filtering $\rho$ and reduction pair
$(\succsim,\succ)$, it is also applicable with argument filtering
$\pi$ and the reduction pair given by: $\aterm \succsim \bterm$ if
$\sigma(\aterm) \succsim \sigma(\bterm)$ and $\aterm \succ \bterm$ if
$\sigma(\aterm) \succ \sigma(\bterm)$.

Now let us prove the theorem.
Suppose $(\Pairs,\Rules,\flag)$ is not finite, so
there is an infinite $(\Pairs,\Rules)$-chain $\rijtje{(\ell_i,r_i,
\gamma_i) \mid i \in \N}$.
If $\flag = \minimal$ then, using Lemma~\ref{lem:usablechain},
we find $\delta_1,\delta_2,\ldots$ such that always
$\rijtje{(\filter(\ell_i),\filter(r_i),\delta_i) \mid i \in \N}$
is an infinite chain,
where each $\filter(r_i)\delta_i \arrr{\filter(U)}
\filter(\ell_{i+1})\delta_{i+1}$, for $U := \userules(\Pairs,\Rules,
\pi) \cup \Ce$ (because $\filter(\Ce) = \Ce$).
Alternatively, if $\flag = \all$, then by Lemma~\ref{lem:filterreduce}
the original chain gives rise to a filtered chain
$\rijtje{(\filter(\ell_i),\filter(r_i),\delta_i) \mid i \in \N}$
where $\delta_i = \gamma^{\filter}$ and each $\filter(r_i)\delta_i
\arrr{\filter(U)} \filter(\ell_{i+1})\delta_{i+1}$ for $U := \Rules$.

Using Lemma~\ref{lem:changechain}, we can alter this chain to a
sequence
$\rijtje{(\filter(\ell_i),\filter(r_i), \epsilon_i) \mid i \in \N}$
which uses the same dependency pairs, and where
$\filter(r_i)\epsilon_i \arrr{\formrules(\filter(\Pairs),\filter(U))}
\filter(\ell_{i+1})\epsilon_{i+1}$.  If a reduction pair as described
exists, then this chain cannot use the filtered pairs
from $\Pairs^\succ$ infinitely often; thus, the original chain must
also have a tail with all dependency pairs in $\Pairs \setminus
\Pairs^\succ$.
\qed
\end{proof}

\begin{proof}[Proof of Lemma~\ref{lem:redtorulesalternative}]
The original reduction $\aterm \arrr{\Rules} \ell\asub$ gives rise to
a formative reduction over $A_\Rules$: all non-collapsing rules of
$\Rules$ are in $\noncollapsing$, and all collapsing rules are
obtained as a concatenation of steps in $\collapsing$ (because the
input of the lemma is an MTRS, so $\FV(r) \subseteq \FV(\ell)$ for
all $\ell \arrz r \in \Rules$).

So, we assume given a formative $\ell$-reduction over $A_\Rules$, and
prove with induction first on the number of non-collapsing steps in
the reduction, second on the length of the reduction, third on the
size of $\aterm$, that $\aterm \arrr{\sformrules(\ell,A_\Rules)}
\ell\asub$ by a formative $\ell$-reduction.
If $\aterm = l\asub$ we are done immediately.
If the reduction uses no
steps at the root, we are quickly done with
the third induction hypothesis and the fact that if $\ell =
\afun(l_1,\ldots,l_n)$ each $\sformrules(l_i,A_\Rules) \subseteq
\sformrules(\ell,A_\Rules)$. Thus, let us assume the reduction
has the form $\aterm \arrr{A_\Rules} \ell'\bsub \arr{A_\Rules}
r'\bsub = \afun(\bterm_1,\ldots,\bterm_n) \arrr{A_\Rules}
\afun(l_1,\ldots,l_n)\asub = \ell\asub$
where $\ell' \arrz r' \in A_\Rules$, with each $\bterm_i
\arrr{A_\Rules} l_i\asub$ by a formative $l_i$-reduction.

If $\ell' \arrz r' \in \noncollapsing$ we are done immediately: then
$r'$ has the form $\afun(\ldots)$, so $\ell' \arrz r' \in
\sformrules(\ell,\Rules)$, and we quickly complete with the first
induction hypothesis.

Thus, let us assume that the reduction has the form
$\aterm \arrr{A_\Rules} \bfun_1(\vec{\avar})\bsub_1 \arr{\collapsing}
\avar_{i_1}\bsub_1 = \bfun_2(\vec{\avar})\bsub_2 \arr{\collapsing}
\avar_{i_2}\bsub_2 = \ldots \arr{\collapsing} \avar_{i_k}\bsub_k =
r'\bsub \arrr{A_\Rules} l\asub$; we know this for $k = 1$, but allow
greater $k$ for more convenient inductive reasoning.  We assume $k$
is chosen as high as possible, so the reduction $\aterm
\arrr{A_\Rules} \bfun_1(\vec{\avar})\bsub_1$ does not end with a
collapsing step at the root.

Now, suppose $\aterm \arrr{A_\Rules} \bfun_1(
\vec{\avar})$ uses no steps at the root,
so $\aterm = \bfun_1(\aterm_1,\ldots,\aterm_{n_1})$ and each
$\aterm_i \arrr{A_\Rules} \avar_i\bsub_1$ by a formative
$\avar_i$-reduction.  Since each $\avar_i$ is a variable, these
reductions are empty: $\aterm = \bfun_1(\vec{\avar})\bsub_1
\arrr{\collapsing} r'\bsub$, and as 
$r'\bsub
\arrr{\sformrules(\ell,A_\Rules)} \ell\asub$ by a formative
$\ell$-reduction, we are done
(all collapsing rules are formative).

Alternatively, suppose the reduction has the form $\aterm
\arrr{A_\Rules} \ell''\csub \arr{A_\Rules} r''\csub \arrr{A_\Rules}
\bfun_1(\vec{\avar})\bsub_1$, with the (formative) reduction $r''\csub
\arrr{\Rules} \bfun_1(\vec{\avar})\bsub_1$ not using any topmost
steps.  Then as in the previous case, we know that $r''\csub =
\bfun_1(\vec{\avar})\bsub_1$.  We cannot have $\ell'' \arrz r'' \in
\collapsing$, by the assumption that $k$ is maximal.  Thus, $r'' =
\bfun_1(r_1,\ldots,r_{n_1})$ and we have $\ell'' \arrz r'' \in
\noncollapsing$; we also have $\bfun_1(\vec{\avar}) \arrz \avar_{i_1}
\in \collapsing$.  Therefore, $X$ (the set constructed in
Definition~\ref{def:combrules}) also contains the rule $\ell'' \arrz
r_{i_1}$.

There are two options.  If $r_{i_1}$ is not a variable, then $\ell''
\arrz r_{i_1} \in \noncollapsing$.  Thus, we can
alter the reduction to have the form $\aterm \arrr{A_\Rules}
\ell''\csub \arr{\noncollapsing} \bfun_2(\vec{\avar})\bsub_2
\arrr{\collapsing} r\bsub \arrr{A_\Rules} \ell\asub$.  This reduction,
which is still a formative $\ell$-reduction, has as many non-collapsing
steps as the original, but one less collapsing step, so we complete
with the second induction hypothesis.

Alternatively, if $r_{i_1}$ is a variable, then $\ell''
\arrr{\collapsing} r_{i_1}$, so we can transform the reduction into
the formative reduction $\aterm \arrr{A_\Rules} \ell''\csub
\arrr{\collapsing} r\bsub \arrr{A_\Rules} \ell\asub$, which uses
fewer non-collapsing rules.  We complete with the first induction
hypothesis.
\qed
\end{proof}

\begin{proof}[Proof of Theorem~\ref{thm:altfilter}]
As in Theorem~\ref{thm:redpairprocessor}, it suffices if a
$(\Pairs,\Rules)$-chain can be transformed into a
$(\filter(\Pairs),\sformrules(\filter(\Pairs),A_{\filter(U)}))$-chain
which uses the same dependency pairs infinitely often.
That is, given a $(\Pairs,\Rules)$-chain $\rijtje{(\ell_i \arrz r_i,
\gamma_i) \mid i \in \N}$, there are $\delta_1,\delta_2,\ldots$ where
that each $\filter(r_i)\delta_i \arrr{\sformrules(\filter(\Pairs),
A_{\filter(U)}))} \filter(\ell_{i+1})\delta_{i+1}$.

To this end, we use either Lemma~\ref{lem:usablechain} or
Lemma~\ref{lem:filterreduce} to find $\epsilon_1,\epsilon_2,\ldots$
which give the $(\filter(\Pairs),\filter(U))$-chain
$\rijtje{(\filter(\ell_i) \arrz \filter(r_i),\epsilon_i) \mid i \in
\N}$.  Lemma~\ref{lem:changechain} gives us $\delta_1,\delta_2$ such
that $\rijtje{(\filter(\ell_i) \arrz \filter(r_i),\delta_i)}$ is a
\emph{formative} $(\filter(\Pairs),\filter(U))$-chain.
Since, by Lemma~\ref{lem:redtorulesalternative}, also
$\filter(r_i)\delta_i \arrr{\sformrules(\filter(\Pairs),A_{\filter(U)
})} \filter(\ell_{i+1})\delta_{i+1}$, we are done.
\qed
\end{proof}

\newpage
\section{Formative Rules with $\TCapp$}
\label{app:tcap}

The definition of usable rules presented in Definition~\ref{def:ur} is an early
definition, which is convenient for presentation, but it is not the
most powerful definition of usable rules in existence.  Later
definitions use the function
$\TCapp$~\cite[Definition~11]{gie:thi:sch:05:1}:

\begin{definition}[$\TCapp$]
Let $\Rules$ be a TRS (an MTRS), let $t$ be a term.
Then
$\TCapp(\aterm,\Rules) = \afun(\TCapp(\aterm_1),\ldots,\TCapp(\aterm_n))$
if $\aterm=\afun(\aterm_1,\ldots,\aterm_n)$ and moreover
$\afun(\TCapp(\aterm_1),\ldots,\TCapp(\aterm_n))$
does not unify with any left-hand side of a rule from $\Rules$;
otherwise $\TCapp(\aterm,\Rules)$ is a different fresh variable for every
occurrence of $\aterm$.
\end{definition}

This definition is used as follows for usable
rules~\cite[Definition 15]{gie:thi:sch:05:1}:

\begin{definition}[Usable Rules with $\TCapp$]\label{def:urtcap}
Let $\bterm$ be a term, $\Rules$ a set of rules and $\pi$ an argument
filtering.
$\tcapuserules(\bterm,\Rules,\pi)$ is the smallest set $\subseteq
\Rules$ such that:
\begin{itemize}
\item if $\Rules$ is not finitely branching,
  then $\tcapuserules(\bterm,\Rules,\pi) = \Rules$;
\item if $\bterm = \afun(\bterm_1,\ldots,\bterm_n)$, then:
  \begin{itemize}
  \item $\tcapuserules(\bterm_i,\Rules,\pi) \subseteq
    \tcapuserules(\bterm,\Rules,\pi)$
    for all $i \in \pi(\afun)$;
  \item any rule $\ell \arrz r \in \Rules$ is in
    $\tcapuserules(\bterm,\Rules,\pi)$ if
    $\afun(\TCapp(\bterm_1),\ldots,\TCapp(\bterm_n))$ unifies with
    $\ell$;
  \end{itemize}
\item if $\ell \to r \in \tcapuserules(\bterm,\Rules,\pi)$, then
  $\tcapuserules(r,\Rules,\pi) \subseteq \tcapuserules(\bterm,\Rules,
  \pi)$.
\end{itemize}
For a set of rules $\Pairs$, we define
$\tcapuserules(\Pairs,\Rules,\pi) = \bigcup_{\aterm \to \bterm \in
\Pairs} \tcapuserules(\bterm,\Rules,\pi)$.
\end{definition}

In all the results in this paper, we could use $\tcapuserules$
exactly like $\userules$, as Lemma~\ref{lem:usablechain} goes through
with this alternative definition (which we will see in
Appendix~\ref{sec:wrt}).

Similar to usable rules, we can also use $\TCapp$ in the definition
of formative rules.  This gives a new formative rules approximation,
which is more powerful than the default ``compare root symbol''
version $\baseformrules$ (which, however, is more suitable for
presentation).

\begin{definition}[Formative Rules with $\TCapp$]
\label{def:formative_impl:2}
Let $\bterm$ be a term and $\Rules$ a set of rules.
$\implformrules(\bterm,\Rules)$ is the smallest set $\subseteq
\Rules$ such that:
\begin{itemize}
\item if $\bterm$ is not linear, then
  $\implformrules(\bterm,\Rules) = \Rules$;
\item if $\bterm = \afun(\bterm_1,\ldots,\bterm_n)$ then:
  \begin{itemize}
  \item $\implformrules(\bterm_i,\Rules) \subseteq
    \implformrules(\bterm,\Rules)$ for all $i$;
  \item if $\bterm : \asort$ then $\ell \arrz \avar \in
    \implformrules(\bterm,\Rules)$ for all collapsing rules
    $\ell \arrz \avar : \asort \in \Rules$;
  \item any rule $\ell \arrz \afun(r_1,\ldots,r_n)$ is in
    $\implformrules(\bterm,\Rules)$ if
    $\afun(\TCapp(r_1),\ldots,\TCapp(r_n))$ unifies with $\bterm$;
  \end{itemize}
\item if $\ell \to r \in \implformrules(\bterm,\Rules)$, then
$\implformrules(\ell,\Rules) \subseteq \implformrules(\bterm,\Rules)$.
\end{itemize}
For a set of rules $\Pairs$, we define $\implformrules(\Pairs,\Rules) =
\bigcup_{\aterm \to \bterm \in \Pairs}
\implformrules(\aterm,\Rules)$.
\end{definition}

Of course, we need to prove that also $\implformrules$ defines a
formative rules approximation.
To this end, we first observe that $\TCapp$ has various nice
properties, as listed in the following lemma:\pagebreak

\begin{lemma}[$\TCapp$ properties]\label{lem:tcapprop}
\begin{enumerate}
\item\label{tcap:instantiate}
  For every $\aterm$ there is some $\asub$ on domain $\FV(\TCapp(\aterm))$
  such that $\aterm = \TCapp(\aterm)\asub$.
\item\label{tcap:both}
  If $\bterm = \TCapp(\aterm)\asub$, then there is some substitution
  $\bsub$ with $\TCapp(\bterm) = \TCapp(\aterm)\bsub$.
\item\label{tcap:reduce}
  If $\TCapp(\aterm)\asub \arr{\Rules} \bterm$, then we can write
  $\bterm = \TCapp(\aterm)\bsub$ for some substitution $\bsub$.
\end{enumerate}
\end{lemma}

\begin{proof}
Each of these statements follows easily with induction on the size of
$\aterm$.

\textbf{(\ref{tcap:instantiate}):} If $\TCapp(\aterm)$ is a variable,
then let $\asub = [\TCapp(\aterm):=\aterm]$.  Otherwise,
$\aterm = \afun(\aterm_1,\ldots,\aterm_n)$ and $\TCapp(\aterm) =
\afun(\TCapp(\aterm_1),\ldots,\TCapp(\aterm_n))$ and there are
$\asub_1,\ldots,\asub_n$ such that each $\aterm_i = \TCapp(\aterm_i)
\gamma_i$.  Since by definition of $\TCapp$ the variables in all
$\TCapp(\aterm_i)$ are disjoint, we can safely define $\gamma :=
\gamma_1 \cup \ldots \cup \gamma_n$.

\textbf{(\ref{tcap:both}):} If $\TCapp(\aterm)$ is a variable, then
we let $\bsub = [\TCapp(\aterm):=\TCapp(\bterm)]$.
Otherwise, $\aterm = \afun(\aterm_1,\ldots,\aterm_n)$ and $\TCapp(
\aterm)=\afun(\TCapp(\aterm_1),\ldots,\TCapp(\aterm_n))$.
Thus we can write $\bterm = \afun(\bterm_1,\ldots,\bterm_n)$ with
each $\bterm_i = \TCapp(\aterm_i)\gamma$.  By the induction
hypothesis, and using again the uniqueness of variables in $\TCapp(
\aterm)$, we thus have $\TCapp(\bterm_i) = \TCapp(\aterm_i)\delta$
for some $\delta$.
If $\TCapp(\bterm) = \afun(\TCapp(\bterm_1),\ldots,\TCapp(\bterm_n))$
we are therefore done.  Otherwise, we must have that
$\afun(\TCapp(\seq{\bterm}))\epsilon = \ell\eta$ for some
left-hand side $\ell$ of $\Rules$ and substitutions $\epsilon,\eta$.
But then also $\afun(\TCapp(\seq{\aterm}))$ unifies with $\ell$,
contradiction with $\TCapp(\aterm)$ not being a variable!

\textbf{(\ref{tcap:reduce}):} If $\TCapp(\aterm)$ is a variable, then
let $\delta = [\TCapp(\aterm):=\bterm]$.  Otherwise, we have $\aterm
= \afun(\aterm_1,\ldots,\aterm_n)$ and $\TCapp(\aterm) =
\afun(\TCapp(\aterm_1),\ldots,\TCapp(\aterm_n))$ does not unify with
any left-hand side.  Thus, the reduction cannot occur at the root; we
can write $\bterm = \afun(\bterm_1,\ldots,\bterm_i,\ldots,\bterm_n)$
with $\TCapp(\aterm_i)\asub \arr{\Rules} \bterm_i$ and
$\TCapp(\aterm_j)\asub = \bterm_j$ for $j \neq i$.  Let each
$\delta_j$ be the restriction of $\gamma$ to variables in
$\TCapp(\aterm_j)$, and let $\delta_i$ be given by the induction
hypothesis, so $\bterm_i = \TCapp(\aterm_i)\delta_i$.  Then $\bterm =
\TCapp(\aterm)\delta$.
\qed
\end{proof}

With these preparations, we can easily obtain the required result:

\begin{theorem}
\label{thm:semimpl:2}
If $\aterm \arrr{\Rules} \ell\gamma$ by a formative
$\ell$-reduction, then this reduction uses only rules in
$\implformrules(\ell,\Rules)$.
Thus, $\implformrules$ is a formative rules approximation.
\end{theorem}

\begin{proof}
By induction on the definition of a formative reduction.

If $\ell$ is non-linear, then $\implformrules(\ell,\Rules) = \Rules$, so
this is obvious.

If $\aterm = \ell\gamma$ then the rules do not play a part.

If $\aterm = \afun(\aterm_1,\ldots,\aterm_n)$ and $\ell = \afun(l_1,
\ldots,l_n)$ and each $\aterm_i \arrr{\Rules} l_i\gamma$ by a
formative $l_i$-reduction, then by the induction hypothesis these
reductions only use rules in $\implformrules(l_i,\Rules)$.
Since $\implformrules(l_i,\Rules) \subseteq \implformrules(l,
\Rules)$, we have the desired result.

If $\aterm \arrr{\Rules} \ell'\delta \arr{\Rules} r'\delta =
\afun(\bterm_1,\ldots,\bterm_n) \arrr{\Rules} \afun(l_1,\ldots,l_n)
\gamma = \ell\gamma$, then by the same reasoning the reduction
$r'\delta \arrr{\Rules} \ell\gamma$ uses only formative rules of
$\ell$, and by the induction hypothesis the reduction $\aterm
\arrr{\Rules} \ell'\delta$ uses only formative rules of $\ell'$. We
are done if $\ell' \arrz r'$ is a formative rule of $\ell$.

This is obvious if $r'$ is a variable, since it necessarily has the
same sort as $\ell$ (sorts are preserved under substitution).
If not a variable, so $r' = \afun(r_1,\ldots,r_n)$, then note that
by Lemma~\ref{lem:tcapprop},
(1) $r_i\delta = \TCapp(r_i)\delta'$ for some $\delta'$, and (2)
since $\TCapp(r_i)\delta' \arrr{\Rules} l_i\gamma$ we can write
$l_i\gamma = \TCapp(r_i)\epsilon$ for some substitution $\epsilon$.
Thus, $l_i$ is unifiable with $\TCapp(r_i)$.  By linearity of $\ell$,
and since the variables in $\TCapp$ are all chosen fresh, this means
$\ell$ is unifiable with $\afun(\TCapp(r_1),\ldots,\TCapp(r_n))$.
\qed
\end{proof}

In the next section, we will also introduce $\sformrules$ with
$\TCapp$.

\newpage
\section{Usable and Formative Rules with respect to an Argument Filtering}\label{sec:wrt}

As observed in the text, the way argument filterings are used with
formative rules differs from the way they are used with usable rules.
Essentially: in usable rules we use $\filter(\userules(\Pairs,\Rules,
\pi))$ whereas in formative rules we use $\formrules(\filter(\Pairs),
\filter(\Rules))$.

Now, we could easily have defined $\baseformrules$ in exactly the
same way as $\userules$, taking the argument filtering into account.
The reason we did not do so is twofold:
\begin{itemize}
\item the \emph{semantic} definition of formative rules with respect
  to an argument filtering is somewhat more complicated;
\item $\formrules(\filter(\Pairs),\filter(\Rules)) \subseteq
  \filter(\baseformrules(\Pairs,\Rules,\pi))$ when defined in the
  obvious way.
\end{itemize}

However, this result does not hold when using $\implformrules$ instead
of $\baseformrules$.  When determining $\implformrules(\afun(l_1,
\ldots,l_n),\Rules,\pi)$ we check whether, for rules $\ell \arrz
\afun(r_1,\ldots,r_n)$, each $\TCapp(r_i)$ unifies with $l_i$; also
for $i \notin \pi(\afun)$!

Therefore, let us now introduce a semantic definition of formative
rules with respect to an argument filtering, and show how it can be
used!

\subsection{Semantic Formative Rules with an Argument Filtering}

Before being able to give the desired semantic definition, we will
need some additional terminology.

\begin{definition}\label{def:regarded}
A reduction step $\aterm \arr{\Rules} \bterm$ occurs at a
\emph{regarded position} of $\aterm$ for an argument filtering $\pi$
if:\footnote{Giesl et al.\ give a similar definition of regarded
  positions in \cite{gie:thi:sch:fal:06}.}
\begin{itemize}
\item $\aterm = \ell\gamma$ and $\bterm = r\gamma$ for $\ell \arrz r
  \in \Rules$ and some substitution $\gamma$, \emph{or}
\item $\aterm = \afun(\aterm_1,\ldots,\aterm_i,\ldots,\aterm_n)$,\ 
  $\bterm = \afun(\aterm_1,\ldots,\aterm_i',\ldots,\aterm_n)$ and
  $\aterm_i \arr{\Rules} \aterm_i'$ and $i \in \pi(\afun)$
\end{itemize}
\end{definition}

As might be expected, regarded positions find their use in filtering:

\begin{lemma}\label{lem:regarded:step}
If $\aterm \arr{\Rules} \bterm$ occurs at a regarded
position for $\pi$, then $\filter(\aterm) \arr{\filter(\Rules)}
\filter(\bterm)$.
\end{lemma}

\begin{proof}
By induction on Definition~\ref{def:regarded}.  If $\aterm =
\ell\gamma \arr{\Rules} r\gamma = \bterm$ for some $\ell \arrz r \in
\Rules$, then $\filter(\aterm) = \filter(\ell)\gamma^{\filter}$ by
Lemma~\ref{lem:filtersubstitute}, $\arr{\filter(\Rules)}
\filter(r)\gamma^{\filter} = \filter(r\gamma)$.

Otherwise, $\aterm = \afun(\ldots,\aterm_j,\ldots)$ and $\bterm =
\afun(\ldots,\aterm_j',\ldots)$ and $\aterm_j \arr{\Rules} \aterm_j'$
and $\pi(\afun) = \{i_1,\ldots,i_k\}$ and $j = i_l$ for some $l$; we
have $\filter(\aterm) = \afun_\pi(\filter(\aterm_{i_1}),\ldots,
\filter(\aterm_{i_l}),\ldots,\filter(\aterm_{i_k}))\linebreak
\arr{\filter(\Rules)} \afun_\pi(\filter(\aterm_{i_1}),\ldots,
\filter(\aterm_{i_l}'),\ldots,\filter(\aterm_{i_k}))$ by the induction
hypothesis, $= \filter(\bterm)$.
\qed
\end{proof}

\begin{lemma}\label{lem:regarded:nostep}
If $\aterm \arr{\Rules} \bterm$ does not occur at a regarded
position for $\pi$, then $\filter(\aterm) = \filter(\bterm)$.
\end{lemma}

\begin{proof}
By induction on the definition of reduction.  The reduction cannot
take place at the root, as this is a regarded position.  Thus,
$\aterm = \afun(\ldots,\aterm_i,\ldots)$ and $\bterm = \afun(\ldots,
\aterm_i',\ldots)$ with $\aterm_i \arr{\Rules} \aterm_i'$.  If
$i \in \pi(\afun)$, then the reduction $\aterm_i \arr{\Rules}
\aterm_i'$ does not occur at a regarded position; hence by the
induction hypothesis $\filter(\aterm_i) = \filter(\aterm_i')$, and
therefore $\filter(\aterm) = \afun_\pi(\ldots,\filter(\aterm_i),
\ldots) = \afun_\pi(\ldots,\filter(\aterm_i'),\ldots) =
\filter(\bterm)$.
Otherwise we directly see that $\filter(\aterm) = \filter(\bterm)$,
as the argument where they differ is filtered away.
\qed
\end{proof}

\begin{definition}[Formative Rules with respect to an Argument
  Filtering]
A function $\formrules$ that maps a term $\ell$, a set of rules
$\Rules$ and an argument filtering $\pi$ to a set $\formrules(\ell,
\Rules,\pi) \subseteq \Rules$ is a formative rules approximation
if for all $\aterm$ and $\gamma$:
if $\aterm \arrr{\Rules} \ell\gamma$ by a formative $\ell$-reduction,
then this reduction uses only rules in $\formrules(\ell,\Rules,\pi)$
at regarded positions for $\pi$.

Given an approximation $\formrules$, we define:
$\formrules(\Pairs,\Rules,\pi) = 
\bigcup_{s \to t \in \Pairs} \formrules(s,\Rules,\pi)$.
\end{definition}

To use this extended definition, we have the following result:

\begin{lemma}\label{lem:filteredformativeusage}
Let $\formrules$ be a formative rules approximation.
Given a formative $(\Pairs,\Rules)$-chain $\rijtje{(\ell_i \arrz r_i,
\gamma_i) \mid i \in \N}$ and an argument filtering $\pi$, each
$\filter(r_i)\gamma_i^{\filter} \arrr{F} \filter(\ell_{i+1})
\gamma_{i+1}^{\filter}$, where $F =  \filter(\formrules(\Pairs,
\Rules,\pi))$.
\end{lemma}

\begin{proof}
Combining Lemmas~\ref{lem:regarded:step} and~\ref{lem:regarded:nostep}:
if in a reduction $\aterm_1 \arr{\{\ell_1 \arrz r_1\}} \aterm_2 \arr{
\{\ell_2 \arrz r_2\}} \ldots \arr{\{\ell_{n-1} \arrz r_{n-1}\}}
\aterm_n$ only steps with rules in $T \subseteq \Rules$ occur at
regarded positions, then $\filter(\aterm_1) \arrr{\filter(T)}
\filter(\aterm_n)$.
Thus, by definition of a formative rules approximation and
Lemma~\ref{lem:filtersubstitute},
always $\filter(r_i)\gamma^{\filter} = \filter(r_i\gamma)
\arrr{\filter(\formrules(\Pairs,\Rules,\pi))} \filter(\ell_{i+1}
\gamma_{i+1}) = \filter(\ell_{i+1})\gamma_{i+1}^{\filter}$.
\qed
\end{proof}

These results culminate in the following alternative to
Theorem~\ref{thm:redpairprocessor}:

\begin{theorem}\label{thm:redpairffrfirst}
Let $(\succsim,\succ)$ be a reduction pair and $\pi$ an argument
filtering. The processor which maps a DP problem $(\Pairs,\Rules,
\flag)$ to the following result is sound:
\begin{itemize}
\item $\{(\Pairs \setminus \Pairs^\succ,\Rules,\flag)\}$ if:
  \begin{itemize}
  \item $\filter(\ell) \succ \filter(r)$ for $\ell \arrz r \in \Pairs^\succ$ and
  $\filter(\ell) \succsim \filter(r)$ for $\ell \arrz r \in \Pairs \setminus \Pairs^\succ$;
  \item $\filter(\ell) \succsim \filter(r)$ for $\ell \arrz r \in \formrules(\Pairs,U,\pi)$ \\
    where $U = \Rules$ if $\flag = \all$ and $U = \userules(\Pairs,\Rules,\pitriv) \cup \Ce$ if $\flag = \minimal$;
  \end{itemize}
\item $\{(\Pairs, \Rules, \flag)\}$ otherwise.
\end{itemize}
\end{theorem}

\begin{proof}
Given an infinite $(\Pairs,\Rules)$-chain, we either keep it as it is
(if $\flag = \all$) or use Lemma~\ref{lem:usablechain} \emph{with an
empty filtering} to turn it into a $(\Pairs,U)$-chain (if $\flag =
\minimal$).  Using Lemma~\ref{lem:changechain} we can turn this chain
formative, and then use Lemma~\ref{lem:filteredformativeusage} to
obtain an infinite decreasing $\succsim$-chain, which contradicts
well-foundedness of $\succ$ unless the elements of $\Pairs^\succ$ only
occur finitely often.
\qed
\end{proof}

\subsection{Syntactic Formative Rules with an Argument Filtering}

We can easily extend the definition of $\baseformrules$ to take into
account the argument filtering directly:

\begin{definition}
Let $\bterm$ be a term, $\Rules$ a set of rules and $\pi$ an argument
filtering.  $\baseformrules(\bterm,\Rules,\pi)$ is the smallest set
$\subseteq \Rules$ such that:
\begin{itemize}
\item if $\bterm$ is not linear, then $\baseformrules(\bterm,\Rules,\pi) =
  \Rules$;
\item if $\bterm = \afun(\bterm_1,\ldots,\bterm_n)$, then:
  \begin{itemize}
  \item $\baseformrules(\bterm_i,\Rules,\pi) \subseteq
    \baseformrules(\bterm,\Rules,\pi)$ for all $i$ \fbox{$\in \pi(\afun)$};
  \item $\{ \ell \to r \in \Rules \mid r$ has shape $\afun \} \subseteq
    \baseformrules(\bterm,\Rules,\pi)$;
  \end{itemize}
\item if $\ell \arrz r \in \baseformrules(\bterm,\Rules,\pi)$, then
  $\baseformrules(\ell,\Rules,\pi) \subseteq \baseformrules(\bterm,\Rules,\pi)$.
\end{itemize}
For rules $\Pairs$, let $\baseformrules(\Pairs,\Rules,\pi) =
\bigcup_{s \to t \in \Pairs} \baseformrules(s,\Rules,\pi)$.
\end{definition}

Compared to Definition~\ref{def:form}, the only difference (apart from
the additional argument $\pi$) is that we
added ``$\in \pi(\afun)$'' in the second condition, as \fbox{highlighted};
thus, $\formrules(\Pairs,\Rules) = \formrules(\Pairs,\Rules,\pitriv)$.
We will see that this is still a formative rules approximation as a
consequence of Theorem~\ref{thm:extendedapproximation}; but first we
must substantiate the claim we made before, that this definition does
not really add anything compared to the definition in the
text.

\begin{lemma}
For all $\Pairs,\Rules,\pi$:
$\baseformrules(\filter(\Pairs),\filter(\Rules)) \subseteq
\filter(\baseformrules(\Pairs,\Rules,\pi))$.
\end{lemma}

\begin{proof}
It suffices if for all terms $\aterm$ we have:
if $\bterm \arrz \cterm \in \baseformrules(\filter(\aterm),
\filter(\Rules))$, then there is $\dterm \arrz \eterm \in
\baseformrules(\aterm,\Rules,\pi)$ such that $\bterm =
\filter(\dterm)$ and $\cterm = \filter(\eterm)$.  This we
prove by induction on the definition of
$\baseformrules(\filter(\aterm),\filter(\Rules))$.

First, note that if $\aterm$ is not linear, then
$\filter(\baseformrules(\Pairs,\Rules,\pi)) = \filter(\Rules)$, so
the requirement is definitely satisfied.  So let us henceforth assume
that $\aterm$ is linear.  Since filtering does not duplicate
variables, also $\filter(\aterm)$ is linear.

If $\aterm$ is a variable, then also $\filter(\aterm)$ is; thus,
$\baseformrules(\filter(\Pairs),\filter(\Rules)) = \emptyset =
\filter(\baseformrules(\Pairs,\linebreak
\Rules,\pi))$.  Therefore, let us
assume that $\aterm$ has the form $\afun(\aterm_1,\ldots,\aterm_n)$.
Let $\pi(\afun) = \{i_1,\ldots,i_k\}$, so $\filter(\aterm) =
\afun_\pi(\filter(\aterm_{i_1}),\ldots,\filter(\aterm_{i_k}))$.
If $\bterm \arrz \cterm \in \baseformrules(\filter(\aterm),
\filter(\Rules))$, this is for one of the following reasons:

\begin{itemize}
\item $\bterm \arrz \cterm \in \baseformrules(\filter(\aterm_i),
  \filter(\Rules))$ for some $i \in \pi(\afun)$; by the induction
  hypothesis we can find suitable $\dterm,\eterm$ such that
  $\dterm \arrz \eterm \in \baseformrules(\aterm_i,\Rules,\pi)
  \subseteq \baseformrules(\aterm,\Rules,\pi)$.
\item $\cterm$ has shape $\afun_\pi$; since $\baseformrules(
  \filter(\aterm),\filter(\bterm)) \subseteq \filter(\Rules)$, we
  can find $\dterm,\eterm$ such that $\bterm = \filter(\dterm)$ and
  $\cterm = \filter(\eterm)$.  Now, if $\cterm$ is a variable, then
  so is $\eterm$; if $\cterm$ has root symbol $\afun_\pi$, then
  $\eterm$ has root symbol $\afun$.  Either way, $\eterm$ has shape
  $\afun$, so $\dterm \arrz \eterm \in \baseformrules(\aterm,
  \Rules,\pi)$.
\item $\bterm \arrz \cterm \in \baseformrules(\bterm',
  \filter(\Rules))$ for some $\bterm' \arrz \cterm' \in
  \baseformrules(\filter(\aterm),\filter(\Rules))$.  By the
  induction hypothesis, we can write $\bterm' = \filter(\ell)$ and
  $\cterm' = \filter(r)$ for some $\ell \arrz r \in
  \baseformrules(\aterm,\Rules,\pi)$., and suitable $\dterm,\eterm$
  exist such that $\dterm \arrz \eterm \in \baseformrules(\ell,
  \Rules,\pi)$.  Then also $\dterm \arrz \eterm \in \baseformrules(
  \aterm,\Rules,\pi)$.
\end{itemize}
\qed
\end{proof}

Thus, we indeed do not lose out by using $\baseformrules$ directly on
filtered terms rather than using the extended definition which takes
the argument filtering $\pi$ into account.  However, it \emph{does}
make a relevant difference for the extended definition of
$\implformrules$:

\begin{definition}[Formative Rules with $\TCapp$ and an Argument Filtering]
Let $\bterm$ be a term, $\Rules$ a set of rules and $\pi$ an
argument filtering.
$\implformrules(\bterm,\Rules,\pi)$ is the smallest set $\subseteq
\Rules$ such that:
\begin{itemize}
\item if $\bterm$ is not linear, then
  $\implformrules(\bterm,\Rules,\pi) = \Rules$;
\item if $\bterm = \afun(\bterm_1,\ldots,\bterm_n)$ then:
  \begin{itemize}
  \item $\implformrules(\bterm_i,\Rules,\pi) \subseteq
    \implformrules(\bterm,\Rules,\pi)$ for all $i$ \fbox{$\in \pi(\afun)$};
  \item if $\bterm : \asort$ then $\ell \arrz \avar \in
    \implformrules(\bterm,\Rules,\pi)$ for all collapsing rules
    $\ell \arrz \avar : \asort \in \Rules$;
  \item any rule $\ell \arrz \afun(r_1,\ldots,r_n)$ is in
    $\implformrules(\bterm,\Rules,\pi)$ if
    $\afun(\TCapp(r_1),\ldots,\TCapp(r_n))$ unifies with $\bterm$;
  \end{itemize}
\item if $\ell \to r \in \implformrules(\bterm,\Rules,\pi)$, then
$\implformrules(\ell,\Rules,\pi) \subseteq \implformrules(\bterm,\Rules,\pi)$.
\end{itemize}
For a set of rules $\Pairs$, we define $\implformrules(\Pairs,\Rules.\pi) =
\bigcup_{\aterm \to \bterm \in \Pairs}
\implformrules(\aterm,\Rules,\pi)$.
\end{definition}

Note that for regarded positions in $\implformrules(\afun(\bterm_1,
\ldots,\bterm_n),\Rules,\pi)$ and rules $\ell \arrz \afun(\seq{r})$
we must check that $\afun(\TCapp(r_1),\ldots,\TCapp(r_n))$ unifies
with $\bterm$, so that each $\TCapp(r_i)$ unifies with $\bterm_i$.
This is \emph{not} limited to $i \in \pi(\afun)$, as it would be if
we considered $\implformrules(\filter(\Pairs),\filter(\Rules))$.

\begin{theorem}\label{thm:extendedapproximation}
If $\aterm \arrr{\Rules} \ell\gamma$ by a formative
$\ell$-reduction, then this reduction uses only rules in
$\implformrules(\ell,\Rules)$ at regarded positions.
Thus, $\implformrules$ is a formative rules approximation.
\end{theorem}

\begin{proof}
We follow exactly the proof of Theorem~\ref{thm:semimpl:2}.
The only point where we have to differ is the case where
$\aterm = \afun(\aterm_1,\ldots,\aterm_n)$ and $\ell = \afun(l_1,
\ldots,l_n)$ and each $\aterm_i \arrr{\Rules} l_i\gamma$ by a
formative $l_i$-reduction.  Here we only need to observe that for
sub-reductions with $i \in \pi(\afun)$ only rules in $\formrules(
l_i,\Pairs,\pi) \subseteq \formrules(\ell,\Pairs,\pi)$ are used at
regarded positions; the reductions in $\aterm_i$ with $i \notin
\pi(\afun)$ are not regarded anyway!
\end{proof}

Note that since obviously $\implformrules(\aterm,\Rules,\pi) \subseteq
\baseformrules(\aterm,\Rules,\pi)$, also $\baseformrules$ is a
formative rules approximation.  Consequently, both can be used in
Theorem~\ref{thm:redpairffrfirst}.

Finally, we can also add the implicit filtering to $\sformrules$,
along with $\TCapp$:

\begin{definition}[Split-Formative Rules with $\TCapp$ and an Argument Filtering]
Let $\bterm$ be a term, $\Rules$ a set of rules and $\pi$ an argument
filtering.
$\tcapsformrules(\bterm,A,\pi)$ is the smallest set $\subseteq A$ such
that:
\begin{itemize}
\item if $\bterm$ is not linear, then $\tcapsformrules(\bterm,A,\pi)
  = A$;
\item all collapsing rules in $A$ are included in
  $\tcapsformrules(\bterm,A,\pi)$;
\item if $\bterm = \afun(\bterm_1,\ldots,\bterm_n)$, then:
  \begin{itemize}
  \item $\tcapsformrules(\bterm_i,A,\pi) \subseteq
    \tcapsformrules(\bterm,A,\pi)$ for all $i \in \pi(\afun)$;
  \item any rule $\ell \arrz \afun(r_1,\ldots,r_n)$ where
    $\afun(\TCapp(r_1),\ldots,\TCapp(r_n))$ unifies with $\bterm$ is
    in $\tcapsformrules(\bterm,A,\pi)$;
  \end{itemize}
\item if $\ell \arrz r \in \tcapsformrules(\bterm,A,\pi)$ and
  $r$ is not a variable, then $\tcapsformrules(\ell,A,\pi)
  \subseteq \tcapsformrules(\bterm,A,\pi)$.
\end{itemize}
For a set of rules $\Pairs$, we define $\tcapsformrules(\Pairs,A,\pi) =
\bigcup_{s \to t \in \Pairs} \tcapsformrules(s,A,\pi)$.
\end{definition}

\begin{lemma}\label{lem:tcapsplit}
Given a set of rules $\Rules$ which satisfies the variable condition
(i.e. $\FV(r) \subseteq \FV(\ell)$ for all $\ell \arrz r \in \Rules$),
if $\aterm \arrr{\Rules} \ell\gamma$
by a formative $\ell$-reduction, then $\aterm \arrr{A_\Rules}
\ell\gamma$ by a formative $\ell$-reduction, which only uses the rules
of $\tcapsformrules(\ell,A_\Rules,\pi)$ on regarded positions.
\end{lemma}

\begin{proof}
Following the proof of Lemma~\ref{lem:redtorulesalternative}, the
reduction $\aterm \arrr{\Rules} \ell\gamma$ is turned into a
formative reduction $\aterm \arrr{A_\Rules} \ell\gamma$ with the
following property: a non-collapsing step at any given position is
never followed by a collapsing step at the same position.

Using induction on the definition of a formative chain, using that
the chain has this property, the lemma follows easily.
\qed
\end{proof}

\begin{theorem}\label{thm:powerproc:split}
Let $(\succsim,\succ)$ be a reduction pair and $\pi$ an argument
filtering.
The processor which maps a DP problem $(\Pairs,\Rules,\flag)$
to the following result is sound:
\begin{itemize}
\item $\{(\Pairs \setminus \Pairs^\succ,\Rules,\flag)\}$ if:
  \begin{itemize}
  \item $\filter(\ell) \succ \filter(r)$ for $\ell \arrz r \in \Pairs^\succ$ and
  $\filter(\ell) \succsim \filter(r)$ for $\ell \arrz r \in \Pairs \setminus \Pairs^\succ$;
  \item $\filter(\ell) \succsim \filter(r)$ for $\ell \arrz r \in \tcapsformrules(\Pairs,A_U,\pi)$, \\
    where $U = \Rules$ if $\flag = \all$ and $U = \tcapuserules(\Pairs,\Rules,\pitriv) \cup \Ce$ if $\flag = \minimal$;
  \end{itemize}
\item $\{(\Pairs, \Rules, \flag)\}$ otherwise.
\end{itemize}
\end{theorem}

\begin{proof}
Let an infinite $(\Pairs,\Rules)$-chain $\rijtje{(\ell_i \arrz r_i,
\gamma_i) \mid i \in \N}$ be given.
If $\flag = \minimal$ then let for all $i$ the substitution $\delta_i
:= \gamma_i$.  Otherwise let $\delta_i$ be given by
Lemma~\ref{lem:usablechain}, with $\pitriv$ as the argument filtering;
the proof to this lemma is easily adapted for $\tcapuserules$ (we do
not present the adaptation here, although an extended version of the
proof, with a different kind of filtering, is given in
Appendix~\ref{subsec:userules:alternative}).

Either way, we obtain an infinite $(\Pairs,U)$-chain
$\rijtje{(\ell_i \arrz r_i,\delta_i) \mid i \in \N}$.
Using Lemma~\ref{lem:changechain}, we can turn this into a formative
$(\Pairs,U)$-chain $\rijtje{(\ell_i \arrz r_i,\epsilon_i) \mid i \in
\N}$.

By Lemma~\ref{lem:tcapsplit}, each reduction $r_i\epsilon_i
\arrr{U} \ell_{i+1}\epsilon_{i+1}$ can be altered to use only rules
in $\tcapsformrules(\ell,A_\Rules,\pi)$ on regarded positions.
Consequently, the filtered chain $\rijtje{(\filter(\ell_i) \arrz
\filter(r_i),\epsilon_i^{\filter}) \mid i \in \N}$ uses only rules in
$\filter(\tcapsformrules(\Pairs,A_\Rules,\pi))$.
The requirements on $(\succsim,\succ)$ guarantee that a tail of this
chain uses no dependency pairs in $\filter(\Pairs^\succ)$!
\qed
\end{proof}

\subsection{$\TCapp$-Conscious Filtering}\label{subsec:tcapconscious}

Unfortunately, when using Theorem~\ref{thm:redpairffrfirst}, we do
lose out in some ways compared to Theorem~\ref{thm:redpairprocessor}:
rather than limiting rules to those in $\userules(\Pairs,\Rules,\pi)$,
we consider $\userules(\Pairs,\Rules)$.  The problem is that the proof
of Theorems~\ref{thm:redpairprocessor} and~\ref{thm:redpairffrfirst}
both rely on transforming the input chain to use only rules in
$U$.  And unfortunately, Lemma~\ref{lem:usablechain} only gives a
transformation to a \emph{filtered} chain.

Towards countering 
this objection,
we will define an alternative way
to ``filter'' terms: $\afilter$.  In
Appendix~\ref{subsec:userules:alternative} we will see how to use this
alternative filtering
to update the proof for usable rules, and obtain the desired processor
(which is also the one we use for implementation).

\begin{definition}\label{def:addtop}
Fixing a set of rules $\Rules$, we introduce for
any sort $\asort$ the fresh symbol $\top_\asort : \asort$, and let
for all terms $\aterm$ the term $\addtop(\aterm)$ be $\TCapp(\aterm)$
with all variables $\avar : \asort$ replaced by $\top_\asort$.
Additionally, let $\Rules_\top$ be a set of rules containing, for all
sorts $\asort$, a rule $\top_\asort \arrz \avar$ with $\avar$ a
variable.
\end{definition}

Obviously, $\Rules_\top$ is non-terminating.  But it does have certain
desirable qualities.

\begin{lemma}\label{lem:reduceanything}
For any term $\aterm : \asort$ we have: $\top_\asort
\arr{\Rules_\top} \aterm$.
\end{lemma}

\begin{proof}
Obvious ($\avar$ can be instantiated with anything).
\qed
\end{proof}

\begin{lemma}\label{lem:addtopsubstitute}
For any term $\aterm$ and substitutions $\gamma$ and $\delta$,
we have:
$\addtop(\aterm)\gamma \arrr{\Rules_\top} \addtop(\aterm\delta)$.
\end{lemma}

\begin{proof}
Let $\aterm : \asort$; we use induction on the size of $\aterm$.

If $\TCapp(\aterm)$ is a variable, then $\addtop(\aterm)\gamma =
\top_\asort\gamma = \top_\asort$.  By Lemma~\ref{lem:reduceanything}
this reduces in one step to $\addtop(\aterm\delta)$.

Otherwise, we can write $\aterm = \afun(\aterm_1,\ldots,\aterm_n)$
where $\afun(\TCapp(\aterm_1),\ldots,\TCapp(\aterm_n))$ does not unify
with any left-hand side of a rule in $\Rules$.  Thus,
$\addtop(\aterm)\gamma = \afun(\addtop(\aterm_1)\gamma,\ldots,
\addtop(\aterm_n)\gamma) \arrr{\Rules_\top}
\afun(\addtop(\aterm_1\delta),\ldots,\addtop(\aterm_n\delta))$ by the
induction hypothesis.  We are done if this is exactly
$\addtop(\afun(\aterm_1\delta,\ldots,\aterm_n\delta))$, which is the
case if $\afun(\TCapp(\aterm_1\delta),\ldots,\TCapp(\aterm_n\delta))$
also does not unify with any left-hand side.  But by nature of
$\TCapp$, this term can be expressed as an instance of $\afun(\TCapp(
\aterm_1),\ldots,\TCapp(\aterm_n))$, so clearly cannot be unified with
anything $\afun(\TCapp(\aterm_1),\ldots,\TCapp(\aterm_n))$ cannot be
unified with!
\qed
\end{proof}

\begin{lemma}\label{lem:addtopsubstitute2}
For any term $\aterm$: $\addtop(\aterm) \arrr{\Rules_\top} \aterm$.
\end{lemma}

\begin{proof}
Observe that we can write $\aterm = \TCapp(\aterm)\delta$ for some
substitution $\delta$ by Lemma~\ref{lem:tcapprop}, and
$\addtop(\aterm)$ can be written as $\TCapp(\aterm)\epsilon$, where
$\epsilon$ is the substitution mapping all variables in $\aterm$ to
the corresponding $\top_\asort$ symbol.  Since always
$\epsilon(\avar) \arr{\Rules_\top} \bterm$ for any $\bterm$, and in
particular for $\bterm = \delta(\avar)$, we have $\addtop(\aterm) =
\TCapp(\aterm)\epsilon \arrr{\Rules_\top} \TCapp(\aterm)\delta =
\aterm$.
\qed
\end{proof}

\begin{lemma}\label{lem:addtopreduce}
If $\aterm \arr{\Rules} \bterm$ then $\addtop(\aterm)
\arrr{\Rules_\top} \addtop(\bterm)$.
\end{lemma}

\begin{proof}
By Lemma~\ref{lem:tcapprop}(\ref{tcap:reduce}), $\TCapp(\bterm) =
\TCapp(\aterm)\gamma$ for some $\gamma$.  Writing $\delta_\top$ for the
substitution that maps each $\avar : \asort$ in $\TCapp(\aterm)$ to
$\top_\asort$, and $\epsilon_\top$ for the substitution that does the
same for $\TCapp(\bterm)$, note that by Lemma~\ref{lem:reduceanything},
each $\delta_\top(\avar) \arr{\Rules} \gamma(\avar)\epsilon_\top$.
Hence we have: $\addtop(\aterm) = \TCapp(\aterm)\delta_\top
\arrr{\Rules_\top} \TCapp(\aterm)\gamma\epsilon_\top =
\addtop(\bterm)$.
\qed
\end{proof}

Using $\addtop$, we can define an alternative way to``filter'' a term.

\begin{definition}
Given an argument filtering $\pi$, the function $\afilter$ is given
by:
\begin{itemize}
\item $\afilter(\avar) = \avar$ for $\avar$ a variable
\item $\afilter(\afun(\aterm_1,\ldots,\aterm_n)) =
  \afun(\aterm_1',\ldots,\aterm_n')$, where each
  $\aterm_i' = \afilter(\aterm_i)$ if $i \in \pi(\afun)$ and
  $\aterm_i' = \addtop(\aterm_i)$ otherwise.
\end{itemize}
For a rule $\ell \arrz r$, let $\afilter(\ell \arrz r)$ be defined as
$\ell \arrz \afilter(r)$ (so with the left-hand side \emph{not}
filtered!).
For a set of rules $\Rules$, let $\afilter(\Rules) = \{ \afilter(\rho)
\mid \rho \in \Rules \}$.
\end{definition}

The normal filtering can be concatenated to $\afilter$:

\begin{lemma}\label{lem:filterafilter}
For all terms $\aterm$: $\filter(\afilter(\aterm)) = \filter(\aterm)$.
\end{lemma}

\begin{proof}
By induction on the definition of $\aterm$.
For a variable it is obvious.
For $\aterm = \afun(\aterm_1,\ldots,\aterm_n)$ with $\pi(\afun) =
\{i_1,\ldots,i_k\}$, we have $\filter(\afilter(\aterm)) =
\filter(\afun(\aterm_1',\ldots,\aterm_n'))$ with $\aterm_j' =
\afilter(\aterm_j)$ if $j \in \pi(\afun)$.  This is equal to
$\afun_\pi(\filter(\afilter(\aterm_{i_1})),\ldots,\filter(\afilter(
\aterm_{i_k})))$, which by the induction hypothesis is exactly
$\afun_\pi(\filter(\aterm_{i_1}),\ldots,\filter(\aterm_{i_k})) =
\filter(\aterm)$.
\qed
\end{proof}

For a corresponding result to the standard $\filter(\aterm)\gamma^{
\filter} = \filter(\aterm\gamma)$, we will use the following two
lemmas:

\begin{lemma}\label{lem:filterleftsubstitute}
For all terms $\aterm$ and substitutions $\gamma$ on a domain
containing $\FV(\aterm)$: $\afilter(\aterm\gamma) \arrr{\Rules_\top}
\aterm\gamma^{\afilter}$.
\end{lemma}

\begin{proof}
By induction on the form of $\aterm$.

If $\aterm$ is a variable, then $\afilter(\aterm\gamma)
\afilter(\gamma(\aterm)) = \gamma^{\afilter}(\aterm) =
\aterm\gamma^{\afilter}$.

Otherwise, $\aterm = \afun(\aterm_1,\ldots,\aterm_n)$ and
$\afilter(\aterm\gamma) = \afun(\bterm_1,\ldots,\bterm_n)$ where
$\bterm_i = \afilter(\aterm_i\gamma)$ if $i \in \pi(\afun)$ and
$\bterm_i = \addtop(\aterm_i\gamma)$ otherwise.

By the induction hypothesis, for $i \in \pi(\afun)$ we have
$\bterm_i = \afilter(\aterm_i\gamma) \arrr{\Rules_\top}
\aterm_i\gamma^{\afilter}$.
For $i \notin \pi(\afun)$ we have $\bterm_i = \addtop(\aterm_i
\gamma) \arrr{\Rules_\top} \aterm_i\gamma$ by
Lemma~\ref{lem:addtopsubstitute2}

Consequently, each $\bterm_i \arrr{\Rules_\top} \aterm_i\gamma$, so
$\afilter(\aterm\gamma) = \afun(\seq{\bterm}) \arrr{\Rules_\top}
\afun(\seq{\aterm})\gamma = \aterm\gamma$.
\qed
\end{proof}

\begin{lemma}\label{lem:filterrightsubstitute}
For all terms $\aterm$ and substitutions $\gamma$ whose domain
contains $\FV(\aterm)$:
$\afilter(\aterm)\gamma^{\afilter} \arrr{\Rules_\top}
\afilter(\aterm\gamma)$.
\end{lemma}

\begin{proof}
By induction on the form of $\aterm$.

If $\aterm$ is a variable, then $\afilter(\aterm)\gamma^{\afilter} =
\gamma^{\afilter}(\aterm) = \afilter(\gamma(\aterm)) =
\afilter(\aterm\gamma)$.

Otherwise, $\aterm = \afun(\aterm_1,\ldots,\aterm_n)$ and
$\afilter(\aterm) = \afun(\aterm_1',\ldots,\aterm_n')$ with
$\aterm_i' = \afilter(\aterm_i)$ for $i \in \pi(\afun)$ and
$\aterm_i' = \addtop(\aterm_i)$ otherwise.  Moreover,
$\afilter(\aterm\gamma) = \afun(\bterm_1,\ldots,\bterm_n)$ where
$\bterm_i = \afilter(\aterm_i\gamma)$ for $i \in \pi(\afun)$ and
$\bterm_i = \addtop(\aterm_i\gamma)$ otherwise.
By the induction hypothesis, for $i \in \pi(\afun)$ we have
$\aterm_i'\gamma^{\afilter} \arrr{\Rules_\top} \bterm_i$.
By Lemma~\ref{lem:addtopsubstitute} the same holds for $i \notin
\pi(\afun)$.
\qed
\end{proof}

With this, we can see that ``filtered'' terms can be reduced with
``filtered'' rules, just like we could do in full filtering (except
that sometimes $\Rules_\top$ rules are needed to fill in the gaps):

\begin{lemma}\label{lem:afilter:reduce}
If $\aterm \arr{\Rules} \bterm$, then $\afilter(\aterm) \arrr{\afilter(
\Rules) \cup \Rules_\top} \afilter(\bterm)$.
\end{lemma}

\begin{proof}
By induction on the derivation of $\aterm \arr{\Rules} \bterm$.

If this is done by a top reduction, so $\aterm = \ell\gamma$ and
$\bterm = r\gamma$ for some $\ell \arrz r \in \Rules$ and substitution
$\gamma$, then $\afilter(\aterm) \arrr{\Rules_\top} \ell
\gamma^{\afilter}$ by Lemma~\ref{lem:filterleftsubstitute}.
This term reduces by $\arr{\afilter(\Rules)}$ to $\afilter(r)
\gamma^{\afilter}$.
By Lemma~\ref{lem:filterrightsubstitute}, this reduces by
$\arrr{\Rules_\top}$ to $\afilter(r\gamma) = \afilter(\bterm)$.

Otherwise, we can write $\aterm = \afun(\ldots,\aterm_i,\ldots)$ and
$\bterm = \afun(\ldots,\bterm_i,\ldots)$ with $\aterm_i \arr{\Rules}
\bterm_i$.  If $i \notin \pi(\afun)$ then observe that
$\addtop(\aterm_i) \arrr{\Rules_\top} \addtop(\bterm_i)$ by
Lemma~\ref{lem:addtopreduce}.  Otherwise $\filter(\aterm_i)
\arrr{\afilter(\Rules) \cup \Rules_\top} \filter(\bterm_i)$ by the
induction hypothesis.  Either way we are done.
\qed
\end{proof}

The important result, to use $\implformrules$ with ``filtered'' rules,
is the following:

\begin{lemma}\label{lem:afilteruseful}
For all $\aterm,\Rules,\pi$:
$\implformrules(\aterm,\afilter(\Rules),\pi) \subseteq
\afilter(\implformrules(\aterm,\Rules,\pi))$.
\end{lemma}

\begin{proof}
We prove: if $\ell \arrz r \in \implformrules(\aterm,\afilter(\Rules),
\pi)$, then some $r'$ exists such that $r = \afilter(r')$ and $\ell
\arrz r' \in \implformrules(\aterm,\Rules,\pi)$ (so therefore
$\afilter(\ell \arrz r') = \ell \arrz r \in
\afilter(\implformrules(\aterm,\Rules,\pi))$).

Obviously, linearity is the same in both cases, so if
$\implformrules(\aterm,\afilter(\Rules),\pi) = \afilter(\Rules)$ by
linearity, the same holds for $\afilter(\implformrules(\aterm,\Rules,
\pi))$.  Otherwise, let $\aterm$ be linear, and assume that $\ell
\arrz r \in \implformrules(\aterm,\afilter(\Rules),\pi)$.
Since $\implformrules(\aterm,\afilter(\Rules),\pi)$ is non-empty,
$\aterm$ cannot be a variable; we can write $\aterm = \afun(\aterm_1,
\ldots,\aterm_n)$.

First, suppose $\ell \arrz r \in \implformrules(\bterm,\afilter(
\Rules),\pi)$ for some $\bterm \arrz \cterm \in \implformrules(
\aterm,\afilter(\Rules),
\pi)$.  By the induction
hypothesis, we can find suitable $\cterm'$ such that $\bterm \arrz
\cterm' \in \implformrules(\aterm,\Rules,\pi)$, and $\cterm =
\afilter(\cterm')$, so applying again the induction hypothesis with
$\ell \arrz r$ and $\bterm$, we find suitable $r'$ such that $\ell
\arrz r' \in \implformrules(\bterm,\Rules,\pi) \subseteq
\implformrules(\aterm,\Rules,\pi)$.

Second, suppose
$\ell \arrz r \in \implformrules(\aterm_i,\afilter(\Rules),\pi)$ for
some $i \in \pi(\afun)$.  Then by the induction hypothesis, we can
find suitable $r'$ with $\ell \arrz r' \in \implformrules(\aterm_i,
\Rules,\pi) \subseteq \implformrules(\aterm,\Rules,\pi)$.

Otherwise, note that since $\ell \arrz r \in \afilter(\Rules)$ we can
always find $r'$ such that $\afilter(\ell \arrz r') = \ell \arrz r$.
If $\ell \arrz r \in \implformrules(\aterm,\afilter(\Rules),\pi)$
because it is collapsing, then $r'$ is a variable too, and therefore
$\ell' \arrz r' \in \implformrules(\aterm,\Rules,\pi)$.

If not collapsing, $r' = \afun(r_1',\ldots,r_n')$ and $r = \afun(r_1,\ldots,
r_n)$ where $r_i = \afilter(r_i')$ for $i \in \pi(\afun)$ but
$r_i = \addtop(r_i')$ and $i \notin \pi(\afun)$; also,
$\afun(\TCapp(r_1),\ldots,\TCapp(r_n))$ unifies with $\aterm$.  This
exactly means that each $\TCapp(r_i)$ unifies with
$\aterm_i$.
We are done if also $\TCapp(r_i')$ unifies with $\aterm_i$.

For $i \notin \pi(\afun)$, this means: given that $\TCapp(\addtop(r_i
'))$ unifies with $\aterm_i$, we must see that also $\TCapp(r_i')$
unifies with $\aterm_i$.  Observing that the $\top_\asort$ symbols
obviously unify with a left-hand side in $\Rules_\top$, we have
$\TCapp(\addtop(r_i')) = \TCapp(r_i')$, so this is obviously true.

For $i \in \pi(\afun)$, this means: given that $\TCapp(\afilter(r_i')
)$ unifies with $\aterm_i$, we must see that also $\TCapp(r_i')$
unifies with $\aterm_i$.  But with induction on the form of $\cterm$
we very easily see that $\TCapp(\afilter(\cterm)) = \TCapp(\cterm)$
for all $\cterm$, modulo renaming of variables.  So this, too, is
obvious.
\qed
\end{proof}

\subsection{Usable Rules with an Argument Filtering: Alternative Way}
\label{subsec:userules:alternative}

As announced in Appendix~\ref{subsec:tcapconscious}, we will now look
at an alternative way to prove Theorem~\ref{thm:userules}, which also
uses the strengthened usable rules from Definition~\ref{def:urtcap}.

\begin{theorem}[Extended version of Theorem~\ref{thm:userules}]\label{thm:userulesext}
Let $(\succsim,\succ)$ be a reduction pair and $\pi$ an argument
filtering.  The processor which maps
$(\Pairs,\Rules,\flag)$
to the following result is sound:
\begin{itemize}
\item $\{(\Pairs \setminus \Pairs^\succ,\Rules,\minimal)\}$ if $\flag = \minimal$ and:
  \begin{itemize}
  \item $\filter(\ell) \succ \filter(r)$ for $\ell \arrz r \in \Pairs^\succ$ and
  $\filter(\ell) \succsim \filter(r)$ for $\ell \arrz r \in \Pairs \setminus \Pairs^\succ$;
  \item $\filter(u) \succsim \filter(v)$ for $u \arrz v \in \tcapuserules(\Pairs,\Rules,\pi) \cup \Ce$.
  \end{itemize}
\item $\{(\Pairs, \Rules, \flag)\}$ otherwise
\end{itemize}
\end{theorem}

We prove this with a reasoning much like the one used in
Appendix~\ref{app:proofs}.  However, we will use $\afilter$ instead
of $\filter$.

\begin{definition}
Fixing a finitely branching set of rules $\Rules$, a set of rules
$\Pairs$ and an argument filtering $\pi$, we introduce for any sort
$\asort$ additional fresh symbols $\bot_\asort : \asort$ and
$\Fc_\asort : [\asort \times \asort] \decpijl \asort$.
The function $\encode$
from \emph{terminating} terms to terms is inductively defined by:
\begin{itemize}
\item if $\aterm$ is a variable, then $\encode(\aterm) = \aterm$;
\item if $\aterm$ has the form $\afun(\aterm_1,\ldots,\aterm_n)$,
  with $\afun(\seq{\aterm}) : \asort$, then let for all $1 \leq k
  \leq n$: $\aterm_k' = \encode(\aterm_k)$ if $k \in \pi(\afun)$ and
  $\aterm_k' = \addtop(\aterm_k)$ otherwise.  We have:
  \begin{itemize}
  \item if $\afun(\TCapp(\aterm_1),\ldots,\TCapp(\aterm_n))$ does not
    unify with any left-hand side of a rule from $\Rules \setminus
    \tcapuserules(\Pairs,\Rules,\pi)$, then $\encode(\aterm) =
    \afun(\aterm_1',\ldots,\aterm_n')$
  \item otherwise,
    $\encode(\aterm) = \Fc_\asort(\afun(\aterm_1',\ldots,\aterm_n'),
    \makelist_\asort(\{\bterm \mid \afun(\seq{\aterm}) \arr{\Rules}
    \bterm\}))$, where $\makelist_\asort$ is defined by:
    \begin{itemize}
    \item $\makelist_\asort(\emptyset) = \bot_\asort$
    \item $\makelist_\asort(X) = \Fc_\asort(\encode(\bterm),
      \makelist_\asort(X \setminus \{\bterm\}))$ if $X$ is non-empty
      and $\bterm$ is its smallest element (lexicographically).
    \end{itemize}
  \end{itemize}
\end{itemize}
\end{definition}

This definition is well-defined because the system is finitely
branching (so $X$ is always finite) and terminating.

Let us start off with some technical lemmas about the relation
between $\addtop$ and $\encode$.

\begin{lemma}\label{lem:addtop:encode}
For all terms $\aterm$: $\addtop(\aterm) \arrr{\Rules_\top}
\encode(\aterm)$.
\end{lemma}

\begin{proof}
By induction on the form of $\aterm$.

If $\addtop(\aterm) = \top_\asort$ for some $\asort$, then we are
immediately done by Lemma~\ref{lem:reduceanything}.  Otherwise,
$\aterm = \afun(\aterm_1,\ldots,\aterm_n)$, with
$\afun(\TCapp(\aterm_1),\ldots,\TCapp(\aterm_n))$ not unifying with
any left-hand side of $\Rules$.
In particular, it does not unify with a left-hand side from $\Rules
\setminus \tcapuserules(\Pairs,\Rules,\pi)$.  Therefore,
$\encode(\aterm) = \afun(\aterm_1',\ldots,\aterm_n')$ where:
\begin{itemize}
\item for $i \in \pi(\afun)$: $\aterm_i' = \encode(\aterm_i)$, so by
  the induction hypothesis $\addtop(\aterm_i) \arrr{\Rules_\top}
  \aterm_i'$;
\item for $i \notin \pi(\afun)$: $\aterm_i' = \addtop(\aterm_i)$, so
  $\addtop(\aterm_i) \arrr{\Rules_\top} \aterm_i'$ in 0 steps.
\end{itemize}
Thus, $\addtop(\aterm) = \afun(\addtop(\aterm_1),\ldots,\addtop(
\aterm_n)) \arrr{\Rules_\top} \afun(\aterm_1',\ldots,\aterm_n') =
\encode(\aterm)$.
\qed
\end{proof}

\begin{lemma}\label{lem:addtop:encode:sub}
For all terms $\aterm$ and substitutions $\gamma$:
$\addtop(\aterm\gamma) \arrr{\Rules_\top} \aterm\gamma^{\encode}$.
\end{lemma}

\begin{proof}
By induction on the form of $\aterm$.

If $\aterm$ is a variable, then $\addtop(\aterm\gamma) =
\addtop(\gamma(\aterm)) \arrr{\Rules_\top} \encode(\gamma(\aterm))
= \aterm\gamma^\encode$ by Lemma~\ref{lem:addtop:encode}.

Otherwise, let $\aterm = \afun(\aterm_1,\ldots,\aterm_n)$.
If $\addtop(\aterm\gamma) = \top_\asort$ for some $\asort$, then we
are immediately done by Lemma~\ref{lem:reduceanything}.  If not,
then $\addtop(\aterm\gamma) = \afun(\addtop(\aterm_1\gamma),\ldots,\linebreak
\addtop(\aterm_n\gamma)) \arrr{\Rules_\top} \afun(\aterm_1
\gamma^\encode,\ldots,\aterm_n\gamma^\encode) = \aterm\gamma^\encode$
by the induction hypothesis.
\qed
\end{proof}

Furthermore, it will be useful to know what happens to the special
$\top_\asort$ symbols, as we will eventually want to get rid of them
again!

\begin{lemma}\label{lem:filterencode}
If $\aterm$ is terminating, and all its (direct and indirect) reducts
do not contain any symbols $\top_\asort$, then
$\filter(\encode(\aterm))$ also does not.
\end{lemma}

\begin{proof}
By induction on $\aterm$ with $\arr{\Rules} \cup\; \rhd$, where $\rhd$
is the strict superterm relation (this union is well-founded because
$\aterm$ is assumed terminating).

If $\aterm$ is a variable, then $\filter(\encode(\aterm)) = \aterm$.

If $\aterm = \afun(\aterm_1,\ldots,\aterm_n)$, then either it does not
unify with the left-hand side of some rule in $\Rules \setminus
\tcapuserules(\Pairs,\Rules,\pi)$ or it does.

In the first case,
  $\encode(\aterm) = \afun(\aterm_1',\ldots,\aterm_n')$ where
  $\aterm_i' = \encode(\aterm_i)$ for $i \in \pi(\afun)$; writing
  $\pi(\afun) = \{ i_1,\ldots,i_k \}$, we have
  $\filter(\encode(\aterm)) = \afun_\pi(\aterm_{i_1}',\ldots,
  \aterm_{i_k}') = \afun_\pi(\encode(\aterm_{i_1}),\ldots,
  \encode(\aterm_{i_k}))$, which by the induction hypothesis does
  not contain any $\top_\asort$-symbols.

In the second case,
  $\encode(\aterm) = \Fc(\afun(\aterm_1',\ldots,\aterm_n'),
  \makelist(X))$ where, as we saw before, $\filter(\afun(\aterm_1',
  \ldots,\aterm_n'))$ does not contain any illegal symbols, and
  $\top_\asort$ can only occur in $\filter(\makelist(X))$ if it
  occurs in some $\filter(\encode(\bterm))$ with $\aterm \arr{\Rules}
  \bterm$.  By the induction hypothesis, this is not the case.
\qed
\end{proof}

Now, let us move on to lemmas corresponding to the results in
Appendix~\ref{app:preliminaries}.

\begin{lemma}\label{lem:encode:left}
For all terms $\aterm$ and substitutions $\gamma$, we have:
$\encode(\aterm\gamma) \arrr{\Ce \cup \Rules_\top}
\aterm\gamma^\encode$.
\end{lemma}

\begin{proof}
By induction on the size of $\aterm$.

For a variable we easily see that $\encode(\aterm\gamma) =
\encode(\gamma(\aterm)) = \gamma^\encode(\aterm) =
\aterm\gamma^\encode$.

If $\aterm = \afun(\aterm_1,\ldots,\aterm_n)$, then let for all
$1 \leq i \leq n$: $\bterm_i = \encode(\aterm_i\asub)$ if $i \in
\pi(\afun)$ and $\bterm_i = \addtop(\aterm_i\asub)$ otherwise.
Then either $\encode(\aterm\asub) = \afun(\seq{\bterm})$, or
$\encode(\aterm\asub) \arr{\Ce} \afun(\seq{\bterm})$.

Now, for all $i$, we have $\bterm_i \arrr{\Ce \cup \Rules_\top}
\aterm_i\gamma^\encode$: if $i \in \pi(\afun)$ then this is provided
by the induction hypothesis, otherwise it holds by
Lemma~\ref{lem:addtop:encode:sub}.
Thus, $\encode(\aterm\asub) \arr{\Ce}^= \afun(\bterm_1,\ldots,
\bterm_n) \arrr{\Ce \cup \Rules_\top} \afun(\aterm_1'\gamma^\encode,
\ldots,\aterm_n'\gamma^\encode) = \aterm\gamma^\encode$.
\qed
\end{proof}

To handle the updated definition of usable rules (with $\TCapp$), we
must also redefine ``completely $\pi$-usable'':

\begin{definition}
A term $\afun(\aterm_1,
\ldots,\aterm_n)$ is \emph{completely $\TCapp$-$\pi$-usable} if:
\begin{itemize}
\item $\afun(\TCapp(\aterm_1),\ldots,\TCapp(\aterm_n))$ does not
  unify with the left-hand side of any rule in $\Rules \setminus
  \tcapuserules(\Pairs,\Rules,\pi)$
\item for all $i \in \pi(\afun)$ also $\aterm_i$ is completely
  $\TCapp$-$\pi$-usable.
\end{itemize}
In addition, all variables are completely $\TCapp$-$\pi$-usable.
\end{definition}

\begin{lemma}\label{lem:encode:right}
For all substitutions $\gamma$ on domain $\FV(\aterm)$ and completely
$\pi$-usable terms $\aterm$:
$\afilter(\aterm)\gamma^\encode \arrr{\Rules_\top} \encode(\aterm\gamma)$.
\end{lemma}

\begin{proof}
By induction on the size of $\aterm$.

For a variable we easily see that
$
\afilter(\aterm)\gamma^\encode =
\aterm\gamma^\encode =
\gamma^\encode(\aterm) =
\encode(\gamma(\aterm)) =
\encode(\aterm\gamma)
$.

If $\aterm = \afun(\aterm_1,\ldots,\aterm_n)$, then
$\afilter(\aterm) = \afun(\aterm_1',\ldots,\aterm_n')$ where
each $\aterm_i' = \afilter(\aterm_i)$ if $i \in \pi(\afun)$ and
$\aterm_i' = \addtop(\aterm_i)$ otherwise.
On the other hand, since $\aterm$ is completely
$\TCapp$-$\pi$-usable, we
have $\encode(\aterm\gamma) = \afun(\bterm_1,\ldots,\bterm_n)$ with
$\bterm_i = \encode(\aterm_i\asub)$ if $i \in \pi(\afun)$ and
$\bterm_i = \addtop(\aterm_i\asub)$ otherwise.
We are done if each $\aterm_i'\gamma^\encode \arrr{\Rules_\top}
\bterm_i$.

For $i \in \pi(\afun)$, note that $\aterm_i$ is also completely
$\TCapp$-$\pi$-usable, so we can use the induction hypothesis:
$\aterm_i'\gamma^\encode = \afilter(\aterm_i)\gamma^\encode
\arrr{\Rules_\top} \encode(\aterm_i\gamma) = \bterm_i$.
For other $i$, we need: $\addtop(\aterm_i)\gamma^\encode
\arrr{\Rules_\top} \addtop(\aterm_i\gamma)$.  This follows
immediately with Lemma~\ref{lem:reduceanything}.
\qed
\end{proof}

\begin{lemma}\label{lem:encode:reduce}
If $\aterm$ is a terminating term and $\Rules$ finitely branching, and
$\aterm \arr{\Rules} \bterm$,
then $\encode(\aterm) \arrr{\afilter(\tcapuserules(\Pairs,\Rules,\pi))
\cup \Ce \cup \Rules_\top} \encode(\bterm)$.
\end{lemma}

\begin{proof}
By induction on the derivation of $\aterm \arr{\Rules} \bterm$. 
Since $\aterm$ reduces, it is not a variable, so let
$\aterm = \afun(\aterm_1,\ldots,\aterm_n)$.
For brevity, also denote $U := \afilter(\tcapuserules(\Pairs,\Rules,
\pi))$.

If $\afun(\TCapp(\aterm_1),\ldots,\TCapp(\aterm_n))$ unifies with the
left-hand side of some rule in $\Rules \setminus U$, then
$\encode(\aterm) \arr{\Ce} \makelist_\asort(X)$, where $X$ is a set
which contains $\bterm$.
It is easy to see that therefore $\makelist_\asort(X) \arrr{\Ce}
\encode(\bterm)$.  Thus, let us henceforth assume that $\afun(
\TCapp(\seq{\aterm}))$ does not unify with any left-hand sides of
$\Rules \setminus U$.  We can write $\encode(\aterm) = \afun(
\aterm_1',\ldots,\aterm_n')$ with $\aterm_i' = \encode(\aterm_i)$ if
$i \in \pi(\afun)$ and $\aterm_i' = \addtop(\aterm_i)$ otherwise.

First suppose the reduction is not at the root, so $\bterm =
\afun(\aterm_1,\ldots,\bterm_i,\ldots,\aterm_n)$ with $\aterm_i
\arr{\Rules} \bterm_i$.  In this case, $\TCapp(\bterm)$ is an
instance of $\TCapp(\aterm)$ by Lemma~\ref{lem:tcapprop}, so does
not unify with anything $\TCapp(\aterm)$ does not unify with.  Thus,
also $\encode(\bterm) = \afun(\aterm_1',\ldots,\bterm_i',\ldots,
\aterm_n')$ with $\bterm_i' = \encode(\bterm_i)$ if $i \in
\pi(\afun)$ and $\bterm_i' = \addtop(\bterm_i)$ otherwise.
In the first case, the induction hypothesis provides that $\aterm_i'
\arrr{U \cup \Ce \cup \Rules_\top} \bterm_i'$.  In the second case,
Lemma~\ref{lem:addtopreduce} gives that $\aterm_i' \arrr{\Rules_\top}
\bterm_i'$.

Alternatively, if the reduction is at the root, then $\aterm =
\ell\gamma$ and $\bterm = r\gamma$ for some rule $\ell \arrz r$ and
substitution $\gamma$.  Since $\aterm = \TCapp(\aterm)\delta$ for
some $\delta$ by Lemma~\ref{lem:tcapprop}, obviously $\TCapp(\aterm)$
unifies with $\ell$, so $\ell \arrz r$ must be in $U$.  Now we
have: $\encode(\aterm) = \encode(\ell\gamma) \arrr{\Ce \cup
\Rules_\top} \ell\gamma^\encode$ by Lemma~\ref{lem:encode:left},
$\arrr{\afilter(U)} \afilter(r)\gamma^\encode
\arrr{\Rules_\top} \encode(r\gamma)$ by Lemma~\ref{lem:encode:right}
because for a usable rule $\ell \arrz r$, indeed $r$ is
completely $\TCapp$-$\pi$-usable!
\qed
\end{proof}

\begin{lemma}\label{lem:furchain}
If there is an infinite minimal $(\Pairs,\Rules)$-chain
$\rijtje{(\ell_i \arrz r_i,\gamma_i) \mid i \in \N}$, then there are
$\delta_1,\delta_2,\ldots$ such that $\rijtje{(\ell_i \arrz
\afilter(r_i),\delta_i) \mid i \in \N}$ is an infinite
$(\afilter(\Pairs),\afilter(\tcapuserules(\Pairs,\Rules,\pi)) \cup
\Ce \cup \Rules_\top)$-chain.

In this chain, the symbols $\top_\asort$ occur only at filtered
positions; that is, if the reduction for a given $i$ has the form
$\afilter(r_i)\delta_i \arrz^* \cterm \arrz^* \ell_{i+1}
\delta_{i+1}$, then $\filter(\cterm)$ does not contain any symbol
$\top_\asort$.
\end{lemma}

\begin{proof}
This is obvious if $\Rules$ is not finitely branching, taking
$\delta_i = \gamma_i$ for all $i$ and noting that the original
chain is over
a signature
not containing any $\top_\asort,\bot_\asort$ or $\Fc_\asort$.

Otherwise, $\encode$ is defined on terminating terms, such as all
$r_i\gamma_i$ and their reducts.
Thus, for all $i$, we define $\delta_i = \gamma_i^\encode$.  Then we
have: $r_i\gamma \arrr{\Rules} \ell_{i+1}\gamma$ implies that
$\filter(r_i)\delta_i \arrr{\Rules_\top}
\encode(r_i\gamma_i)$ (by Lemma~\ref{lem:encode:right})
$\arrr{\afilter(\tcapuserules(\Pairs,\Rules,\pi)) \cup \Ce \cup
\Rules_\top}
\encode(\ell_i\gamma_{i+1})$ (by Lemma~\ref{lem:encode:reduce})
$\arrr{\Ce \cup \Rules_\top} \ell_i\delta_{i+1}$ (by
Lemma~\ref{lem:encode:left}).

Now we observe that since minimality gives termination of $\Rules$ on
the $r_i\gamma$, that any rule which is used in the reduction
$r_i\gamma_i \arrr{\Rules} \ell_{i+1}\gamma_{i+1}$ satisfies the
variable condition: the right-hand side of the rule contains no
variables not occurring on the left.  Limiting interest to only rules
which actually occur in the original chain, we can therefore say that
if $r_i\gamma_i$ does not contain $\top$-symbols (and we assume that
the initial chain does not), then neither does any of its direct or
indirect reducts.

Observing that
$\filter(\afilter(r_i)\gamma_i^\encode) =
\filter(\afilter(r_i))\gamma_i^{\filter \circ \encode} =
\filter(r_i)\gamma_i^{\filter \circ \encode}$, and noting that that
$\filter(r_i)$ does not contain any $\top_\asort$ symbols because
$r_i$ does not, we also see that $\gamma_i^{\filter \circ \encode}$
(which maps variables to values $\filter(\encode(\gamma(\avar)))$)
does not contain any $\top$-symbols by
Lemma~\ref{lem:filterencode}, using the above observation that no
reducts of $r_i\gamma_i$ (or any of its subterms) contains a
$\top$-symbol.

Thus, the first element of the reduction
$\afilter(r_i)\gamma_i^\encode$ contains the $\top$-symbols only at
filtered positions.  The lemma follows because this property is
preserved by reduction: if $\cterm$ contains $\top_\asort$ only at
filtered positions, and $\cterm \arr{\afilter(\userules(\Pairs,
\Rules,\pi)) \cup \Ce \cup \Rules_\top} \dterm$, then the same holds
for $\dterm$.  This follows easily with induction on $\cterm$ (noting
that the right-hand sides of rules in $\afilter(\userules(\Pairs,
\Rules,\pi))$ only contain $\top$-symbols at filtered positions, and
using that the right-hand sides of rules do not introduce variables
not occurring on the left).
\qed
\end{proof}

\begin{proof}[Proof of Theorem~\ref{thm:userulesext}]
Let $\rijtje{(\ell_i \arrz r_i,\gamma_i) \mid i \in \N}$ be an
infinite minimal $(\Pairs \setminus \Pairs^\succ,\Rules)$-chain.
By Lemma~\ref{lem:furchain}, we can
find $\delta_1,\delta_2,\ldots$ such that
$\rijtje{(\ell_i \arrz \afilter(r_i),\delta_i) \mid i \in \N}$ is a
chain.  This new chain uses (filtered versions of) the same
dependency pairs infinitely often, but uses
$\afilter(\tcapuserules(\Pairs,\Rules,\pi)) \cup \Ce \cup \Rules_\top$
for the reduction $\afilter(r_i)\delta_i \arrz^*
\ell_{i+1}\delta_{i+1}$, and contains the $\Rules_\top$ symbols only
at filtered positions.

Now observe that using the results in Appendix~\ref{app:basics}, we
\pagebreak
can modify this chain with full filterings.  We have, for all $i$:
$\filter(r_i)\delta_i^{\filter} = \filter(\afilter(r_i))
\delta_i^{\filter}$ (by Lemma~\ref{lem:filterafilter}), $=
\filter(\afilter(r_i)\delta_i) \arrr{\filter(\afilter(\tcapuserules(
\Pairs,\Rules,\pi)) \cup \Ce \cup \Rules_\top)} \filter(\afilter(
\ell_{i+1})\delta_{i+1}) = \filter(\ell_{i+1})\delta_{i+1}^{\filter}$.
Moreover, this reduction does not use the rules in
$\filter(\Rules_\top)$, as the $\top$-symbols occur only at filtered
positions in the reduction, and thus are all filtered away.

Since $\filter(\afilter(\tcapuserules(\Pairs,\Rules,\pi))) =
\filter(\tcapuserules(\Pairs,\Rules,\pi))$ and we can assume, as we
did in the proof of Theorem~\ref{thm:userules}, that $\pi$ does not
touch the $\Fc_\asort$ symbols, the requirements guarantee that the
new chain cannot use filtered elements of $\Pairs^\succ$ infinitely
often.  Hence, the original chain must have a tail without these
pairs.
\qed
\end{proof}

\subsection{Usable and Formative Rules with $\TCapp$ and an argument
Filtering}\label{subsec:useformtcapfilter}

And now, at last, do we have all the preparations for the result that
$\implformrules$ and $\tcapuserules$ can be combined as you might
hope, including argument filterings:

\begin{theorem}\label{thm:powerproc}
Let $(\succsim,\succ)$ be a reduction pair and $\pi$ an argument
filtering. The processor which maps a DP problem $(\Pairs,\Rules,
\flag)$ to the following result is sound:
\begin{itemize}
\item $\{(\Pairs \setminus \Pairs^\succ,\Rules,\flag)\}$ if:
  \begin{itemize}
  \item $\filter(\ell) \succ \filter(r)$ for $\ell \arrz r \in \Pairs^\succ$ and
  $\filter(\ell) \succsim \filter(r)$ for $\ell \arrz r \in \Pairs \setminus \Pairs^\succ$;
  \item $\filter(\ell) \succsim \filter(r)$ for $\ell \arrz r \in \implformrules(\Pairs,U,\pi)$ \\
    where $U = \Rules$ if $\flag = \all$ and $U = \tcapuserules(\Pairs,\Rules,\pi) \cup \Ce$ if $\flag = \minimal$;
  \end{itemize}
\item $\{(\Pairs, \Rules, \flag)\}$ otherwise.
\end{itemize}
\end{theorem}

\begin{proof}
Assume given an infinite $(\Pairs,\Rules)$-chain $\rijtje{(\ell_i
\arrz r_i,\gamma_i) \mid i \in \N}$.

If $\flag = \all$, then let, for all $i$, $\delta_i = \gamma_i^{\afilter}$.
Then by
Lemmas~\ref{lem:filterleftsubstitute}--\ref{lem:afilter:reduce},
$\rijtje{(\ell_i \arrz \afilter(r_i),\delta_i)}$ is an
infinite $(\afilter(\Pairs),\afilter(U) \cup \Rules_\top)$-chain,
where $U = \Rules$.  Moreover, it is easy to see that since the
$r_i\gamma_i$ do not contain any $\top_\asort$-symbols, all terms in
this chain contain $\top_\asort$-symbols only at filtered positions.

On the other hand, if $\flag = \minimal$, then we may assume that all
$r_i\gamma_i$ are terminating.  By Lemma~\ref{lem:furchain} there are
$\delta_1,\delta_2,\ldots$ such that
$\rijtje{(\ell_i \arrz \afilter(r_i),\delta_i) \mid i \in \N}$ is an
infinite $(\afilter(\Pairs), \afilter(U) \cup \Rules_\top)$-chain,
where $U = \tcapuserules(\Pairs,\Rules,\pi) \cup \Ce$.  Also in this
chain, the $\top_\asort$-symbols occur only at filtered positions.

Now, using Lemma~\ref{lem:changechain}, we can turn this chain
formative: there are substitutions $\epsilon_1,\epsilon_2,\ldots$
such that $\rijtje{(\ell_i \arrz \afilter(r_i),\epsilon_i) \mid i \in
\N}$ is an infinite formative $(\afilter(\Pairs),\afilter(U) \cup
\Rules_\top)$-chain.  Considering the way this result is proved, no
$\top_\asort$ symbols are introduced at unfiltered positions.

Since $\implformrules$ is a formative rules approximation, we see by
Lemma~\ref{lem:filteredformativeusage} that each
$\filter(r_i)\epsilon_i^{\filter} =
\filter(\afilter(r_i))\epsilon_i^{\filter} \arrr{F}
\filter(\ell_{i+1})\epsilon_i^{\filter}$, where $F = \filter(
\formrules(\Pairs,\afilter(U) \cup \Rules_\top,\pi))$.
And what is more: since the original reduction $\afilter(r_i)
\epsilon_i$ only contains $\top_\asort$-symbols at filtered
positions, this filtered reduction does not contain those symbols at
all!  Thus, the $\Rules_\top$ rules are never used.
In addition, by Lemma~\ref{lem:afilteruseful}, $\filter(\formrules(
\Pairs,\afilter(U),\pi)) \subseteq \filter(\afilter(\formrules(\Pairs,
U,\pi))) = \filter(\formrules(\Pairs,U,\pi))$.
Thus, we obtain an infinite $(\filter(\Pairs),\filter(\formrules(
\Pairs,U,\pi)))$-chain $\rijtje{(\filter(\ell_i) \arrz \filter(r_i),
\epsilon_i^{\filter}) \mid i \in \N}$.

Throughout these transformations, the dependency pairs (in filtered
form) have stayed in the same place.  Since the reduction pair proves
that for $\ell \arrz r \in \Pairs^\succ$ the
pair $\filter(\ell) \arrz \filter(r)$ cannot occur infinitely often,
we therefore also see that the original chain must have had a tail
where pairs in $\Pairs^\succ$ do not occur.
\qed
\end{proof}

\newpage
\section{Implementation Details}
\label{app:aprove_impl}
One might wonder why we went to such lengths to prove a slightly
stronger version of Theorem~\ref{thm:redpairprocessor}; the difference
in $\filter(\implformrules(\Pairs,\Rules,\pi))$ and
$\implformrules(\filter(\Pairs),\filter(\Rules))$ is not so great, and
we even lose out by considering linearity before filtering rather than
after (although we could get around this issue without large
problems).

The main reason for this is implementation: it is significantly easier
to encode, for a given rule $\ell \arrz \afun(r_1,\ldots,r_n)$, that
it is usable if $\afun(\TCapp(r_1),\ldots,\TCapp(r_n))$ unifies with
the left-hand side of a usable rule, than it is to encode that
$\filter(\afun(\TCapp(r_1),\ldots,\TCapp(r_n)))$ unifies with the filtered
left-hand side of a usable rule: at the time of encoding, the
filtering is not yet known. For formative rules, the issue is analogous.

\subsection{Implementing Formative Rules $\implformrules$}

When both formative rules with respect to an argument filtering
and usable rules with respect to an argument filtering
are used in the reduction pair processor on a DP problem
$(\Pairs,\Rules,\minimal)$, \aprove\ orients the following
rules from $\Rules$ with $\succsim$:
\[
(\implformrules(\Pairs,\tcapuserules(\Pairs,\Rules) \cup \Ce,\pi) \: \cap \:
\tcapuserules(\Pairs,\Rules,\pi)) \: \cup \: \Ce
\]
The reason for this particular choice is that it allowed to
reuse significant parts of \aprove's implementation for the
constraint encodings of usable rules with respect to an argument
filtering \cite{cod:gie:sch:thi:12}.
As in \cite{cod:gie:sch:thi:12}, the generated constraints are
then solved by means of SAT or SMT solvers.

To see that
\[
(\implformrules(\Pairs,\tcapuserules(\Pairs,\Rules) \cup \Ce,\pi) \: \cap \:
\tcapuserules(\Pairs,\Rules,\pi)) \: \cup \: \Ce
\]
is a superset of the rules from $\Rules \cup \Ce$
that need to be oriented by $\succsim$
in Theorem~\ref{thm:powerproc} (and hence that this choice leads to a
sound DP processor), it suffices to consider
the following the points.
\begin{enumerate}
\item
We have:
\[
\implformrules(\Pairs,\tcapuserules(\Pairs,\Rules) \cup \Ce,\pi) \supseteq
\implformrules(\Pairs,\tcapuserules(\Pairs,\Rules,\pi) \cup \Ce,\pi)\]
because $\tcapuserules(\Pairs,\Rules) \supseteq
\tcapuserules(\Pairs,\Rules,\pi)$
and by monotonicity of
$\implformrules(\Pairs,U,\pi)$
with respect to the set of rules $U$, i.e.:
\[
 U \subseteq U' \quad \Rightarrow \quad
 \implformrules(\Pairs,U,\pi) \subseteq \implformrules(\Pairs,U',\pi)
\]
\item
Moreover, we have:
\[
\tcapuserules(\Pairs,\Rules, \pi) \cup \Ce \supseteq
\implformrules(\Pairs,\tcapuserules(\Pairs,\Rules,\pi) \cup \Ce,\pi)
\cup \Ce
\]
because
$\implformrules(\Pairs,U, \pi) \subseteq U$.
\end{enumerate}

\subsection{Implementing Split-Formative Rules $\tcapsformrules$}

For split-formative rules, \aprove\ directly follows
Theorem~\ref{thm:powerproc:split}, where we consider only the usable
rules with respect to the trivial argument filtering.
Conveniently (for automation, not for power), they can be
computed in advance.


} 

\end{document}